\documentclass[12pt,colorinlistoftodos]{article}

\usepackage{amsmath,amsfonts,amssymb,amsthm}

\usepackage{arydshln}
\usepackage{graphicx}
\usepackage{subcaption} 
\usepackage{xcolor}

\usepackage{booktabs}
\usepackage{hyperref}
\usepackage{bm}
\usepackage{doi}

\newtheorem{theorem}{Theorem}
\newtheorem{lemma}[theorem]{Lemma}
\newtheorem{conjecture}[theorem]{Conjecture}
\usepackage{geometry}
\usepackage{todonotes}

\def\IR{\mathbb{R}}
\def\bE{\bm{E}}
\def\bF{\mathbf{F}}
\def\bW{\mathbf{W}}
\def\bX{\mathbf{X}}
\def\bu{\mathbf{u}}
\def\bx{\mathbf{x}}
\def\by{\mathbf{y}}
\def\bz{\bm{z}}
\def\b0{\bm{0}}

\def\R{\mathcal{R}}
\def\diag{\mathsf{diag}}

\hypersetup{
    colorlinks=true,
    linkcolor=blue,
    citecolor=blue,
    urlcolor=blue
}

\usepackage{stmaryrd}

\usepackage{tikz}
\usetikzlibrary{shapes,arrows}
\usetikzlibrary{positioning}
\tikzstyle{cloud} = [draw, ellipse,fill=red!20, node distance=0.87cm,
minimum height=2em]
\tikzstyle{line} = [draw, -latex']
\usetikzlibrary{shapes.symbols,shapes.callouts,patterns}
\usetikzlibrary{calc}
\usepackage{dsfont} 

\usepackage{float}

\graphicspath{{./fig/}}

\def\diag{\mathsf{diag}}
\def\dbracket#1{\left\llbracket#1\right\rrbracket}

\title{Transmission of multiple pathogens across species}
\author{Clotilde Djuikem \and Julien Arino}

\begin{document}



\maketitle

\begin{abstract}

We analyse a model that describes the propagation of many pathogens within and between many species. 
A branching process approximation is used to compute the probability of disease outbreaks.
Special cases of aquatic environments with two host species and one or two pathogens are considered both analytically and computationally.

\end{abstract}

\section{Introduction}

The ranges of species are continuously changing \cite{kirkpatrick1997evolution,sexton2009evolution}. 
However, the process has accelerated in recent years because of climate change \cite{atkins2010local,parmesan2006ecological}.
Regardless of what is driving their evolution, a consequence of the modification of ranges is more frequent interactions between species that did not use to interact or interacted quite infrequently.

This has a wide variety of consequences.
Competition for resources is modified if an invading species is, for instance, using the same resource as a resident one.
This is thought to be one of the main drivers of species evolution \cite{phillips2010life}.
Range shifting can also lead to the introduction into ecosystems of pathogens from which they were absent, when species whose range now includes these ecosystems become more frequent there \cite{carlson2022climate}.
Introductions of pathogens due to range shifting is also very similar to what happens when human populations encroach into the ranges of species \cite{ellwanger2021zoonotic}, which has led to an increasing number of spillover events \cite{meadows2023historical}.

In both cases, some of the populations involved may be hosts to a wide variety of pathogens. 
Understanding a situation with different pathogens and different species is therefore important.

The specific motivation for the present work comes from the observation that salmonids are observed increasingly frequently the Mackenzie River, in the western Canadian arctic.
Of interest to collaborators from Fisheries and Oceans Canada (see \emph{Acknowledgments}) is the fact that these vagrant salmon species spend most of their lives in distant ecosystems, where they can acquire pathogens that are, to this point, mostly absent from the Mackenzie River aquatic ecosystem.
When they are collocated in that ecosystem, those vagrant species can in turn transmit those novel pathogens to resident species.

While pathogens abound in terrestrial ecosystems, the situation is even more pronounced in aquatic ecosystems, where numerous pathogens are present \cite{bergh1989high,wommack2000virioplankton,wyn2001enteric}.
Viruses, for instance, are estimated to be the most abundant ``lifeforms'' in the oceans, representing over 90\% of the nucleic-acid-containing particles and about 5\% of the biomass there \cite{suttle2007marine}.
Many aquatic pathogens infect fish species, so that the invading species mentioned earlier may be coming into contact with a wide variety of pathogens prior to their entering a novel ecosystem.

Our aim is therefore to establish models to help understand the introduction of pathogens in species from which they were absent up to that point, when these species come into contact with other species potentially bearing the pathogen.
The model in this paper is a simplified model and serves to set the general setting in which we operate.
We use a simple SLIR model, whose dynamics in a single location and single population is well understood, but assume that there are multiple species of hosts as well as multiple pathogen species.
We also assume that there is no coinfection with multiple pathogens.

The article is organised as follows. 
In Section~\ref{sec:multi-species}, an ordinary differential equations (ODE) multi-species epidemiological model is introduced, followed by its continuous time Markov chain (CTMC) equivalent.
The section also presents an analysis of both the ODE and CTMC models, the latter using a branching process approximation to compute the probability of a disease outbreak.
Section~\ref{sec:2p1v} focuses on the case of two species and one pathogen, with example scenarios corresponding to three different fish viruses investigated numerically, as well as a particular case focusing on introduction of a pathogen by a species in which it is endemic.
In Section~\ref{sec:2p2v}, the case of two species and two pathogens is discussed.

\section{The general model}
\label{sec:multi-species}

Consider $P$ populations. 
These populations could be the same or different species, the important feature being that they be distinguishable according to some criterion.
In the sequel, we use both terms, \emph{species} and \emph{population}, interchangeably.
Within and between these populations, $V$ pathogens can propagate.
Specifically, each population is described by an SLIR epidemic model, where susceptible individuals in a given population can become infected by any of the $V$ viruses if they come into contact with an individual infected by it, regardless of the population that individual belongs to.
Further, we assume that coinfection does not occur, i.e., once infected by a given pathogen, an individual cannot acquire infection from another pathogen.

\subsection{Formulation of the deterministic model}
For $p=1,\ldots,P$, denote $S_p$ the number of individuals susceptible to infection in the $p$th population. 
Individuals of species $p$ may become infected by any of the $V$ pathogens present when they come across an individual infectious with that pathogen.
The parameter describing the rate at which contacts between a susceptible from species $p$ and an infectious individual from species $q=1,\ldots,P$ infected with pathogen $v=1,\ldots,V$, results in new infections, is $\beta_{pqv}$, i.e., in words, the parameter $\beta$ has indices
\begin{equation*}
\beta_{\text{who becomes infected, who infects, with which pathogen}}.
\end{equation*}

Upon infection, individuals of species $p=1,\ldots,P$ infected by pathogen $v=1,\ldots,V$ become latently infected, with numbers denoted $L_{pv}$.
We do not consider coinfection with multiple pathogens; once an individual is contaminated with any of the viruses, they cannot become infected by any other pathogen.

After an incubation period of mean duration $1/\varepsilon_{pv}$ time units, individuals of species $p$ infected with pathogen $v$ become infectious to others. The number of such infectious individuals is denoted $I_{pv}$.
Finally, after an average $1/\gamma_{pv}$ time units spent infectious with the pathogen, individuals recover and move to the $R_p$ compartment. 
At this point, the pathogen they were infected with is ignored as it is not relevant to the problem under consideration.
Regarding species demography, birth into population $p=1,\ldots, P$ occurs at the fixed rate $b_p$, while death occurs in all compartments at the \emph{per capita} rate $d_p$. 

Taking all this into account, we have a group model,  with dynamics of the different states governed for population $p=1,\ldots P$ and pathogen $v=1,\ldots,V$ by the following ordinary differential equations:
\begin{subequations}\label{sys:ODE-multi-species}
    \begin{align}
        &\dot S_p=b_p-\left(\sum_{q=1}^{P}\sum_{v=1}^{V}\beta_{pqv}I_{qv}+d_p \right)S_p, \label{sys:multi-species_dS} \\
        & \dot L_{pv}=\sum_{q=1}^{P}\beta_{pqv}I_{qv} S_p-(\varepsilon_{pv}+d_p)L_{pv}, \label{sys:multi-species_dL}\\
        &\dot I_{pv}=\varepsilon_{pv}L_{pv}-(\gamma_{pv}+d_p)I_{pv}, \label{sys:multi-species_dI}\\
        & \dot R_p=\sum_{v=1}^{V}\gamma_{pv}I_{pv}-d_pR_{p}. \label{sys:multi-species_dR} 
    \end{align}
\end{subequations}
System \eqref{sys:ODE-multi-species} is considered with nonnegative initial conditions. 
To avoid a trivial case, it is assumed that $L_{pv}+I_{pv}>0$ for at least one $(p,v)\in\{1,\ldots,P\}\times\{1,\ldots,V\}$.
The total size of each population $p$ for $i=1,\ldots, P$  is given by:
\begin{equation}\label{eq:total_pop_pop}
N_p=S_p+\sum_{v=1}^{V}L_{pv}+\sum_{v=1}^{V}I_{pv}+R_p.
\end{equation}

\begin{table}
    \centering
    \begin{tabular}{cl}
    \toprule
         Variable & Meaning \\
    \midrule
       $S_p$  &  Susceptible individual in population $p$\\
       $L_{pv}$  & Latent individual in population $p$ infected by virus $v$ \\
       $I_{pv}$  & Individual in population $p$ infectious with virus $v$\\
       $R_p$ & Recovered individual in population $p$ \\
    \bottomrule
    \end{tabular}
    \caption{State variables and their meaning.}
    \label{tab:state_variables}
\end{table}

\subsection{Notation}
Equations \eqref{sys:multi-species_dL} and \eqref{sys:multi-species_dI} involve two different indices. 
Analysis of the system often requires to list these indices.
To simplify presentation, for given symbols $X$ and $Y$, we use the notation
\[
    \dbracket{X_p}=X_1,\ldots,X_P,
\]
\[
    \dbracket{X_{pv}}= X_{11}, X_{12},\ldots, X_{1V}, X_{21},X_{22},\ldots,X_{2V},\ldots,X_{P1},X_{P2},\ldots,X_{PV}
\]
and
\begin{multline*}
    \dbracket{X_{pv}+Y_p}= X_{11}+Y_1, X_{12}+Y_1,\ldots, X_{1V}+Y_1,\\
    X_{21}+Y_2,X_{22}+Y_2,\ldots,X_{2V}+Y_2,\ldots, \\
    X_{P1}+Y_P,X_{P2}+Y_P,\ldots,X_{PV}+Y_P.
\end{multline*}
Thus, when multiple indices are present, we present indices as would the row-first enumeration of indices of the entries of a $P\times V$ matrix.
Note that the assumption is that indices $p$ and $v$ are reserved, respectively, for population and virus species indices and therefore run in $1,\ldots,P$ and $1,\ldots,V$.

\subsection{Basic analysis of the deterministic model}
\label{sec:basic_analysis_ODE}
The disease-free equilibrium (DFE) of system~\eqref{sys:ODE-multi-species} is 
\begin{equation}\label{eq:DFE}
\bE_0^{\eqref{sys:ODE-multi-species}}=
\left(\dbracket{S_p^0},0_{\mathbb{R}^{P(2V+1)}}\right),
\end{equation}
where $S_p^0=b_p/d_p$ for $p=1, \ldots, P$. Note that for equilibria as well as for the basic reproduction number, we use a superscript to refer to the specific form of the system that is being considered.

To determine the matrices used in the computation of the basic reproduction number using the next generation matrix method of \cite{VdDWatmough2002}, order infected variables by type: $\dbracket{L_{pv}},\dbracket{I_{pv}}$.
Then the non-negative $2PV\times 2PV$-matrix $\mathbf{G}$ has block form
\begin{equation}\label{eq:matrix_F}
\mathbf{G}=\begin{bmatrix}
0 & \mathbf{G}_{12} \\
0 & 0 \\
\end{bmatrix},
\end{equation}
where the $PV\times PV$-matrix $\mathbf{G}_{12}$ is itself a block matrix, with each $V\times V$ sized block taking the form, for $p,q\in\{1,\ldots,P\}$,
\begin{equation}\label{eq:matrices_G}
\mathbf{G}_{12}^{pq}=S_p^0\diag(\beta_{pq1},\ldots, \beta_{pqV}).
\end{equation}

The matrix $\mathbf{W}$ is a non-negative $2PV\times 2PV$-matrix and has block form
\begin{equation}\label{eq:matrix_W}
\mathbf{W}=\begin{bmatrix}
\mathbf{W}_{11}& \mathbf{0}\\
-\mathbf{W}_{21}& \mathbf{W}_{22} \\
\end{bmatrix},
\end{equation}
with the $PV\times PV$-sized block taking the form
\begin{equation}
\label{eq:matrices_W}
\mathbf{W}_{11}=\diag(\llbracket\varepsilon_{pv}+d_p\rrbracket), 
\mathbf{W}_{21}=\diag(\llbracket\varepsilon_{pv}\rrbracket)
\text{ and }
\mathbf{W}_{22}=\diag(\llbracket\gamma_{pv}+d_p\rrbracket).    
\end{equation}

The basic reproduction number of \eqref{sys:ODE-multi-species} is then the spectral radius of $\mathbf{GW}^{-1}$ and is given by
$\mathcal{R}_0^{\eqref{sys:ODE-multi-species}}
=\rho(\mathbf{GW^{-1}})$.
Since $\mathbf{W}$ is block lower triangular, 
\[  
\mathbf{W^{-1}}=   
\begin{bmatrix} 
\mathbf{W}_{11}^{-1}& \mathbf{0}\\
\mathbf{W}_{22}^{-1}\mathbf{W}_{21}\mathbf{W}_{11}^{-1}& \mathbf{W}_{22}^{-1}
\end{bmatrix},
\]
whence, from the form of $\mathbf{G}$, we obtain 
\[
\mathcal{R}_0^{\eqref{sys:ODE-multi-species}}=\rho( \mathbf{G}_{12}\mathbf{W}_{22}^{-1}\mathbf{W}_{21}\mathbf{W}_{11}^{-1}).
\]

The matrices \eqref{eq:matrices_W} are diagonal and the structure of $\mathbf{G}_{12}$ is relatively simple; it is therefore possible to further simplify the expression above. We obtain
\begin{equation}\label{eq:R0-multi-scpeies}
    \mathcal{R}_0^{\eqref{sys:ODE-multi-species}}=\max_{v=1,\ldots,V}\mathcal{R}_0^v,
\end{equation}
where
\[
\mathcal{R}_0^v=\rho(\mathcal{B}_v),
\]
is the basic reproduction number of virus $v=1,\ldots, V$ and $\mathcal{B}_v$ is a $P\times P$-matrix defined as
\begin{equation}\label{eq:matrix_Bv}
\mathcal{B}_v=[\mathcal{B}_{pqv}]_{pq},
\end{equation}
where  \(\mathcal{B}_{pqv}\) denotes the element in the \(p\)-th row and \(q\)-th column of the matrix \(\mathcal{B}_v\) and is given by
\[
\mathcal{B}_{pqv}=\frac{\beta_{pqv} \varepsilon_{qv} S_p^0}{(\varepsilon_{qv} +d_{q})(\gamma_{qv}+d_{q})},
\]
for $p,q=1,\ldots,P,\; v=1,\ldots,V$.

From \cite[Theorem 2]{VdDWatmough2002}, we deduce the following result concerning the local asymptotic stability of the disease-free equilibrium $\bE_0^{\eqref{sys:ODE-multi-species}}$.
\begin{lemma}\label{lem:las-multi-species}
    The disease-free equilibrium $\bE_0^{\eqref{sys:ODE-multi-species}}$ of \eqref{sys:ODE-multi-species} is locally asymptotically stable if $\mathcal{R}_0^{\eqref{sys:ODE-multi-species}} < 1$ and unstable if $\mathcal{R}_0^{\eqref{sys:ODE-multi-species}} > 1$.
\end{lemma}

Remark that in the absence of interaction between the populations, i.e., when $\beta_{p\ell v}=0$ if $p\neq\ell$, the basic reproduction number in each population $p=1,\ldots,P$ is given by
\begin{equation}\label{eq:single_R0p}
   \mathcal{R}_{0p}=\max_{v=1,\ldots,V}(\mathcal{R}_{0p}^v),
   \text{ where }
   \mathcal{R}_{0p}^v=\frac{\beta_{ppv} \varepsilon_{pv} S_p^0}{(\varepsilon_{pv}+d_p)(\gamma_{pv}+d_p)}.
\end{equation}
Various forms of the reproduction numbers appear in the remainder of the text. To clarify, we list them here.
\begin{itemize}
\item $\mathcal{R}_{0p}$ denotes the basic reproduction number of species $p$ in the presence of multiple pathogens $V$, excluding other species.
\item $\mathcal{R}_{0p}^v$ denotes the basic reproduction number of species $p$ in the presence of a single pathogen $v$, excluding other species.
\item $\mathcal{R}_{0}^v$ denotes the basic reproduction number of a single pathogen $v$ across $P$ interacting species.
\item $\mathcal{R}_{0}$ denotes the basic reproduction number of multiple pathogens $V$ across multiple interacting species $P$, i.e, for the full System \eqref{sys:ODE-multi-species}.
\end{itemize}

Asymptotic stability is in fact global for a given pathogen when $\R_{0}<1$, as established in the following theorem.
\begin{theorem}\label{th:GAS_R0}
If $\mathcal{R}_{0}^{\eqref{sys:ODE-multi-species}} < 1$, 
then the DFE $\bE_{0}^{\eqref{sys:ODE-multi-species}}$  is globally asymptotically stable (GAS) in $\Omega$, where 
\begin{multline*}
    \Omega = \Biggl\{
    (\llbracket S_p\rrbracket,\llbracket L_{pv}\rrbracket,\llbracket I_{pv}\rrbracket,\llbracket R_p\rrbracket) \in \mathbb{R}^{2P(V+1)} : \\
    N_p=S_P+\sum_{v=1}^{V}(L_{pv}+I_{pv})+R_p \leq \frac{b_p}{d_p}; \quad p=1,\ldots, P
    \Biggr\}.
\end{multline*}
\end{theorem}
This result is proved in Appendix~\ref{app:proof_GAS_R0}.

\subsection{The continuous time Markov chain model}
The ODE model~\eqref{sys:ODE-multi-species} is the limit of a continuous-time Markov chain (CTMC) \cite{Kurtz1970}, which we now consider.
Indeed, CTMCs provide allow to explore various scenarios that the deterministic models cannot capture.
In particular, CTMCs of the type used here have discrete state variables, thereby allowing transitions toward a state where the disease is eradicated.
In ODE models, such states are typically approached only as limits, which leads to implausible situations \cite{fowler2021atto}.

Let $T$ denote the set of all finite-dimensional vectors whose components are nonnegative integers. 
The CTMC model related to the deterministic model~\eqref{sys:ODE-multi-species} then takes the form,
\begin{subequations}\label{sys:CTMC_general}
    \begin{equation}\label{sys:CTMC_general_Xt}
        \mathbf{X}_t=
        \left(
\llbracket S_p(t)\rrbracket,\llbracket L_{pv}(t)\rrbracket,\llbracket I_{pv}(t)\rrbracket,\llbracket R_p(t)\rrbracket
\right),\quad t\in\mathbb{R}_+,
\end{equation}
where each element of the vector is a collection of discrete random variables that take values in $T$ and where the time between events is exponentially distributed \cite{allen2010introduction}.
The process is time-homogeneous as the rates in the ODE are constants, so the CTMC is characterised by transition probabilities from state 
$\mathbf{k}$ to state $\mathbf{j}$,
\begin{equation}\label{sys:CTMC_general_proba}
\mathbb{P}(\mathbf{X}(t+\Delta t)=\mathbf{j} \mid \mathbf{X}(t)=\mathbf{k})=\sigma(\mathbf{k}, \mathbf{j}) ,    
\end{equation}
\end{subequations}
with transition rates $\sigma(\mathbf{k}, \mathbf{j})$ given in Table~\ref{tab:prob-multi-species}.

\begin{table}[htbp]
\centering
\begin{tabular}{lll}
\toprule
Event & Transition  & Transition rate \\
\midrule
Birth of $S_p$ & $S_p\rightarrow S_p+1$ & $b_p$ \\
Natural death of $S_p$ & $S_p\rightarrow S_p-1$ & $d_p S_p$ \\
Natural death of $L_{pv}$ & $L_{pv}\rightarrow L_{pv}-1$ & $d_p L_{pv}$ \\
Natural death of $I_{pv}$ & $I_{pv}\rightarrow I_{pv}-1$ & $d_p I_{pv}$ \\
Natural death of $R_p$ & $R_p\rightarrow R_p-1$ & $d_p R_p$ \\
Infection of $S_p$ by $I_{qv}$ & $S_{p} \rightarrow S_{p}-1, L_{pv} \rightarrow L_{pv}+1$ & $\beta_{pqv} I_{qv} S_{p}$ \\
End of incubation of $L_{pv}$  & $L_{pv} \rightarrow L_{pv}-1, I_{pv} \rightarrow I_{pv}+1$ & $\varepsilon_{pv}L_{pv}$ \\
Recovery of $I_{pv}$ & $ I_{pv} \rightarrow I_{pv}+1,  R_{p} \rightarrow R_{p}+1$ & $\gamma_{pv} I_{pv}$ \\
\bottomrule
\end{tabular}  
\caption{Events, transitions $\mathbf{k}\to\mathbf{j}$ and transition rates $\sigma(\mathbf{k}, \mathbf{j})$ of the general CTMC model \eqref{sys:CTMC_general}.}
\label{tab:prob-multi-species}
\end{table}

\subsection{Branching process approximation of the CTMC}

Branching process approximations (BPA) simplify the analysis of CTMC by approximating them with simpler branching processes. By considering only the number of offspring generated by each state transition, BPA reduce the complexity associated with modelling and analysing CTMC. 
We use multitype branching process approximation (MBPA) to approximate the CTMC \eqref{sys:CTMC_general} near the disease-free equilibrium, where the term \emph{multitype} is used to indicate that there are different types of offspring.

The reasoning below follows \cite{ALLEN201399}; see also \cite{GreenwoodGordillo2009,lahodny2013probability}.
Let $Z=(Z^{\ell}, Z^i)$ be a $2PV$-multitype Galton-Watson process that approximates the CTMC $\mathbf{X}_t$, with superscripts \(\ell\) and \(i\) used for latent and infectious states, respectively. 
This MTBA is the homogeneous vector Markov process whose states are vectors in $T$ defined above.

Call \emph{infectious type} $(k,p,v)$ infected individuals of species $s=1,\ldots,P$ who are latent ($k=\ell$) or infectious ($k=i$) with virus $v=1,\ldots,V$.
Using the definition of multitype branching processes given in \cite{athreya1972branching, Haccou_Jagers_Vatutin_2005, Harris1963,  kurtz1972relationships}, the 2PV-multitype Galton-Watson process \(Z(t)=(Z^\ell,Z^i)\) is the vector of discrete random variables $Z_{pv}^k(t),\; k\in\{\ell,i\},\; p=1,\ldots,P,\; v=1,\ldots,V$ for the number of individuals in infectious type \((k,p,v)\) at time \( t \in [0, \infty) \), \( Z_{pv}(t) \in T \).
Let 
\[
F(t, \bu) = \left(\llbracket F_{pv}^\ell(\bu)\rrbracket,\llbracket F_{pv}^i(\bu)\rrbracket\right)
\]
denote the probability generating function (p.g.f.) for the entire process. 

The p.g.f. for type \((k,p,v)\) at time \( t \) is defined as
\[
F_{pv}^k(t, \bu) = \mathbb{E}\left(\bu^{Z(t)} | Z(0) = e_i\right),
\]
where \( e_i \) is the \( i \)-th unit vector, \(  \bu=(\llbracket u^\ell_{pv} \rrbracket, \llbracket u^i_{pv} \rrbracket) \), \( u^k_{pv} \in [0, 1]\) and we denote \( u^{Z(t)} = (u_{11}^\ell)^{Z^L_1(t)} \dots (u_{PV}^i)^{Z^I_{PV}(t)} \). 
The number of offspring produced by an individual in infectious type \((k,p,v)\) is independent of the number of individuals in that infectious type or in infectious type \((k',p',v')\neq(k,p,v)\).
Further, the p.g.f. \(F_{pv}^k(t,\bu)\) is a solution of the backward Kolmogorov differential equation,
\[
\frac{\partial}{\partial t} F_{pv}^k(t,\bu) = \omega_{pv}^k [ f_{pv}^k(F(t,\bu)) - F_{pv}^k(t,u)], \quad p = 1, \dots, ,\; v = 1, \dots, V,
\]
with initial conditions \( F_{pv}^k(0,\bu) = u_{pv}^k \),where \( \omega_{pv}^k \) is the rate parameter for the exponentially distributed lifetime of infectious type \((k,p,v)\) and $f_{pv}^k$ is the offspring p.g.f. for type $(k,p,v)$ given that $Z(0)=e_i$ \cite{athreya1972branching, Haccou_Jagers_Vatutin_2005, Harris1963,  kurtz1972relationships}.

More specifically, the offspring p.g.f. $f_{pv}^k$ for type \((k,p,v)\) is defined as \( f_{pv}^k: [0,1]^n \to [0,1] \),
 \begin{equation}\label{sys:pgf-multi-species-F}
    f(\bu)= \left(\llbracket f_{pv}^\ell(\bu)\rrbracket,\llbracket f_{pv}^i(\bu)\rrbracket\right),
\end{equation}
where for $k\in\{\ell,i\}$,
\begin{equation}\label{gen_form_pgf}
f^k_{pv}(\bu)
=
\sum_{\llbracket r^\ell_{pv}\rrbracket,\llbracket r^i_{pv}\rrbracket=0}^{\infty}
\mathbb{P}(\llbracket r_{pv}^k \rrbracket)\ 
(u_{11}^\ell)^{r^\ell_{11}}\cdots \ (u_{PV}^\ell)^{r^\ell_{PV}} \ (u_{11}^i)^{r^i_{11}}\ \cdots \ (u_{PV}^i)^{r^i_{PV}},
\end{equation}
with $\mathbb{P}(\llbracket r_{pv}^k \rrbracket)=\mathbb{P}(\llbracket r^\ell_{pv} \rrbracket, \llbracket r^i_{pv} \rrbracket)$ the probability that an individual of type \((k,p,v)\) gives birth to $r^k_{j}$ individuals of type \(u_j^k\), for $j \in \{\llbracket{pv} \rrbracket\}$ and $k\in\{\ell,i\}$. 
These probabilities are computed using the infinitesimal transition probabilities of the CTMC \eqref{sys:CTMC_general} in Table~\ref{tab:prob-multi-species} with the approximation $S_p=S_p^0$.  

Our focus is on the process near the disease-free equilibrium.
In this context, since the initial number of infectious individuals is small, the assumption of independence is reasonable. 
Since the process is time-homogeneous, individuals in infectious type \((k,p,v)\) have the same offspring p.g.f. at all times.

Using the general form in \eqref{gen_form_pgf}, the offspring p.g.f. are given by
\begin{subequations}\label{sys:pgf-multi-species}
\begin{align}
   & f_{pv}^\ell(\bu)= \frac{ \varepsilon_{pv} u_{pv}^i+d_p}{\varepsilon_{pv}+d_p}, \label{sys:pgf-multi-species-flpv}\\
   & f_{pv}^i(\bu)=\frac{\left(\displaystyle\sum_{q=1}^{P} \beta_{qpv}S_{q}^0 u_{qv}^\ell\right)u_{pv}^i+\gamma_{pv}+d_p}{\Lambda_{p v}},
   \label{sys:pgf-multi-species-fipv}
\end{align}
\end{subequations}
with
\begin{equation*}
\Lambda_{pv}=\sum_{q=1}^{P}\beta_{pqv}S_{q}^0+\gamma_{pv}+d_p.
\end{equation*}
The following result then holds.
\begin{theorem}\label{th:exists-q-multi-pv}
    The probability of extinction in the multitype branching process with probability generating functions~\eqref{sys:pgf-multi-species} is given by
    \begin{subequations}
    \label{eq:prob-ext-gen2}
    \begin{align}
        \mathbb{P}_{\text{ext}} &= \prod_{p=1}^{P}\prod_{v=1}^{V}\left(\frac{ \varepsilon_{pv} z_{pv}^i+d_p}{\varepsilon_{pv}+d_p}\right)^{\ell_{pv0}}(z_{pv}^i)^{i_{pv0}} ,\\
        \mathbb{P}_{\text{outbreak}} &= 1- \mathbb{P}_{\text{ext}},
    \end{align}    
    \end{subequations}
    where 
    \begin{equation*}\label{eq:def-z}
        \bm{z}:=\left(\llbracket z_{pv}^\ell\rrbracket, \llbracket z_{pv}^i\rrbracket\right)
    \end{equation*}
    is a fixed point on $ \displaystyle[0,1]^{2PV}$ of the p.g.f.~\eqref{sys:pgf-multi-species} and $\llbracket L_{pv}(0)\rrbracket=\llbracket\ell_{pv0}\rrbracket$, $\llbracket I_{pv}(0) \rrbracket=\llbracket i_{pv0}\rrbracket$ is the initial condition. The following alternative holds:
    \begin{itemize}
        \item if $\mathcal{R}_0^{\eqref{sys:ODE-multi-species}} \leq 1$, then $\bm{z} = \mathbf{1}$, i.e., $\mathbb{P}_{\text{ext}}=1$;
        \item if $\mathcal{R}_0^{\eqref{sys:ODE-multi-species}}> 1$, then additionally to $\bz=\bm{1}$, there is a unique vector $\mathbf{0}<\bm{z}< \mathbf{1}$ such that $\bF(\bz)=\bz$.
    \end{itemize}
\end{theorem}

\begin{proof}
The probability of disease extinction or disease outbreak relative to the CTMC $\mathbf{X}_t$ can be approximated by the extinction probability of the MBPA with generating function~\eqref{sys:pgf-multi-species}.
For a small number of latent and infectious individuals and under the assumptions above about the branching process approximation, the probabilities of disease extinction and disease outbreak, given $\llbracket L_{pv}(0)\rrbracket=\llbracket\ell_{pv0}\rrbracket$, $\llbracket I_{pv}(0) \rrbracket=\llbracket i_{pv0}\rrbracket$ and $\mathcal{R}_0^{\eqref{sys:ODE-multi-species}}>1$, are
\begin{subequations}
    \label{eq:prob-ext-gen1}
    \begin{align}
        \mathbb{P}_{\text{ext}} &= \prod_{p=1}^{P}\prod_{v=1}^{V}(z_{pv}^\ell)^{\ell_{pv0}}(z_{pv}^i)^{i_{pv0}} \\
        \mathbb{P}_{\text{outbreak}} &= 1-\mathbb{P}_{\text{ext}},
    \end{align}
\end{subequations}
where $\bm{z}:=\left(\llbracket z_{pv}^\ell\rrbracket, \llbracket z_{pv}^i\rrbracket\right)$ is a fixed point of the p.g.f.~\eqref{sys:pgf-multi-species} on $ \displaystyle[0,1]^{2PV}$ such that $\mathbf{0}\leq \bm{z} < \mathbf{1}$, denoting $\mathbf{0}=\mathbf{0}_{2PV}$ and $\mathbf{1}$ the $2PV$ vector of all ones.
It is clear that \eqref{sys:pgf-multi-species-flpv} implies that $z_{pv}^\ell=(\varepsilon_{pv} z_{pv}^i+d_p)/(\varepsilon_{pv}+d_p)$, where $z_{pv}^i$ is the fixed point of \eqref{sys:pgf-multi-species-fipv}. 
As a consequence, the probabilities in \eqref{eq:prob-ext-gen1} defined for $\mathcal{R}_0^{\eqref{sys:ODE-multi-species}} > 1$ become those in \eqref{eq:prob-ext-gen2}.
Proof of the existence part is deferred to Appendix~\ref{app:proof-exists-q-multi-pv}.
\end{proof}
In MBPA, processes either reach zero or approach infinity. 
The probability of extinction is interpreted in our model as the probability of a minor epidemic, while an outbreak is the establishment of the pathogen.
However, once the number of infected individuals in the branching process reaches a certain level, it is no longer an accurate approximation of the epidemic, as the MBPA does not adequately represent dynamics far from the disease-free equilibrium.

\section{Case of one pathogen and two species}
\label{sec:2p1v}
To get better insight into the behaviour of the system in a tractable case, we consider the case with $P=2$ species and $V=1$ pathogen.
We first specialise the model and results of mathematical analysis to this special case (Section~\ref{subsec:2P1V-analysis}), then consider numerically four specific transmission scenarios, which we also summarise by indicating the type of propagation taking place between populations $P_1$ and $P_2$.
\begin{enumerate}
    \item Pathogen propagation within and between species (Section~\ref{subsec:IHN}, $P_1\leftrightarrow P_2$).
    \item The pathogen is transmitted to both species, but one species can only infect members of its species (Section~\ref{subsec:wild-to-farm}, $P_1\to P_2,P_2\not\to P_1$).
    \item Both species can acquire the pathogen, but one of the two species does not become a transmitter (Section~\ref{subsec:VHS}, $P_1\to P_2,P_2\not\to$). 
    \item The pathogen is established at an endemic level in one species and is absent from the other species (Section~\ref{subsec:intro-from-endemic}, $P_1^\star\to P_2$).
\end{enumerate}

\subsection{The models and their basic analysis}
\label{subsec:2P1V-analysis}

Because $V=1$, the second index of $\varepsilon$ and $\gamma$ and the third index of $\beta$ always equals 1. 
To simplify notation, in the remainder of Section~\ref{sec:2p1v}, we drop this superfluous index and write $\varepsilon_p$, $\gamma_p$ and $\beta_{pq}$, for $p,q=1,\ldots,P$.
\subsubsection{The ODE model when $P=2$ and $V=1$}

Setting $P=2$ and $V=1$ in \eqref{sys:ODE-multi-species} gives 
\begin{subequations} \label{sys:ode-2p1v}
\begin{align}
&\dot S_1=b_1-\beta_{11}  S_1I_1-\beta_{12} S_1  I_2-d_1S_1 \\
&\dot L_1=\beta_{11}  S_1 I_1+\beta_{12} S_1  I_2-(\varepsilon_{1}+d_1)L_1 \\
&\dot I_1=\varepsilon_{1} L_1-(\gamma_{1}+d_1) I_1 \\
&\dot  R_1=\gamma_{1} I_1-d_1 R_1 \\
&\dot  S_2=b_2-\beta_{21} S_2 I_1-\beta_{22} S_2 I_2-d_2S_2\\
&\dot  L_2=\beta_{21} S_2 I_1+\beta_{22} S_2 I_2-(\varepsilon_{2}+d_2) L_2 \\
&\dot I_2=\varepsilon_{2} L_2-(\gamma_{2} +d_2) I_2 \\
&\dot R_2=\gamma_{2} I_2-d_2 R_2,
\end{align}
\end{subequations}

The system is considered with variables ordered as $S_1, S_2, L_1, L_2, I_1, I_2, R_1, R_2$.
The analysis in Section~\ref{sec:basic_analysis_ODE} carries through, taking into account that since $V=1$, a few adaptations of the terms and matrices defined there are required.
The disease-free equilibrium (DFE) of \eqref{sys:ode-2p1v} is 
\begin{equation}\label{eq:DFE-2P1V}
\bE_0^{\eqref{sys:ode-2p1v}}=
\left(\dbracket{S_p^0},\bm{0}_{\mathbb{R}^{3P}}\right)=\left(\frac{b_1}{d_1},\frac{b_2}{d_2},\bm{0}_{\mathbb{R}^6}\right),
\end{equation}
infected variables are $\dbracket{L_{p}},\dbracket{I_{p}}$, $P\times P$ matrix $\mathbf{G}_{12}$ is not a block matrix but instead has entry $(p,q)$ equal to $S_p^0\beta_{pq}$ and, since $P=2$,
\[
   \mathbf{G}_{12}=\begin{pmatrix} 
   \beta_{11}  S_1^0 & \beta_{12}S_1^0 \\
   \beta_{21}  S_2^0 & \beta_{22} S_2^0 
   \end{pmatrix}.
\]
Blocks in the matrix $\mathbf{W}$ take the form
\[
\mathbf{W}_{11}=\diag(\llbracket\varepsilon_{p}+d_p\rrbracket), 
\mathbf{W}_{21}=\diag(\llbracket\varepsilon_{p}\rrbracket)
\text{ and }
\mathbf{W}_{22}=\diag(\llbracket\gamma_{p}+d_p\rrbracket),  
\]
so that in the case $P=2$ under consideration,
\[
   \mathbf{W}=\begin{pmatrix}
   \varepsilon_{1}+d_1 & 0 & 0 & 0\\
   0 & \varepsilon_{2}+d_2 & 0 & 0\\
   - \varepsilon_{1} & 0 & \gamma_{1}+d_1 & 0\\
   0 & - \varepsilon_{2} & 0 & \gamma_{2} +d_2
   \end{pmatrix}.
\]
It follows that the basic reproduction number of \eqref{sys:ode-2p1v} is 
\begin{equation}\label{eq:R0-2P1V}
\mathcal{R}_0^{\eqref{sys:ode-2p1v}}=\frac{\beta_{11} \kappa_{1} S_1^0 +\beta_{22} \kappa_{2} S_2^0+\sqrt{\left(\beta_{11} \kappa_{1} S_1^0 -\beta_{22} \kappa_{2} S_2^0\right)^2+4\beta_{12}\beta_{21}\kappa_{1}\kappa_{2}S_1^0S_2^0}}{2},
\end{equation}
where, for $p=1,2$,
\[
\kappa_{p}=\frac{ \varepsilon_{p} }{(\varepsilon_{p}+d_p)(\gamma_{p}+d_p)}.
\]
Results of Lemma~\ref{lem:las-multi-species} carry forward to the local asymptotic stability or instability of \eqref{eq:DFE-2P1V} based on the value of $\mathcal{R}_0^{\eqref{sys:ode-2p1v}}$ as defined by \eqref{eq:R0-2P1V}.

\subsubsection{The CTMC model when $P=2$ and $V=1$}
In this case of one virus and two species, the CTMC takes the form
\begin{equation}\label{sys:CTMC_2p1v}
    \bX_t=(S_1(t), S_2(t), L_1(t), L_2(t), I_1(t), I_2(t), R_1(t), R_2(t)),\quad t\in\IR_+
\end{equation}
and is characterised by the transition rates in Table~\ref{tab:prob-2p1v}.

\begin{table}[H]
\centering
\begin{tabular}{lll}
\toprule
Event ($p=1,2$) & Transition   & Transition rates \\
\midrule
Birth of $S_p$ & $S_p\rightarrow S_p+1$ & $b_p   $ \\
Natural death of $S_p$ & $S_p\rightarrow S_p-1$ & $d_p S_p  $ \\
Natural death of  $L_p$ & $L_p\rightarrow L_p-1$ & $d_p L_p  $ \\
Natural death of  $I_p$ & $I_p\rightarrow I_p-1$ & $d_p I_p  $ \\
Natural death of  $R_p$ & $R_p\rightarrow R_p-1$ & $d_p R_p  $ \\
Infection of $S_p$ by $I_q$ & $S_p\rightarrow S_p-1, L_p\rightarrow L_p+1$ & $\beta_{pq} S_p I_{q}$ \\
End of incubation of $L_p$ & $L_p\rightarrow L_p-1, I_p\rightarrow I_p+1$ & $\varepsilon_{p} L_p  $ \\
Recovery in  $I_p$ & $I_p\rightarrow I_p-1, R_p\rightarrow R_p+1$ & $\gamma_{p} I_p  $ \\
\bottomrule
\end{tabular}
\caption{Reaction rates used to determine transition probabilities for 2-species and 1-pathogen CTMC model.}
\label{tab:prob-2p1v}
\end{table}

\subsubsection{Branching process approximation}
Theorem~\ref{th:exists-q-multi-pv} is specialised to the $P=2$, $V=1$ case by letting $Z=(L_1, L_2, I_1, I_2)$ be a MBPA of the CTMC defined in \eqref{sys:CTMC_2p1v}, with infected types $\ell_{10}$, $\ell_{20}$, $i_{10}$ and $i_{20}$. 
The p.g.f. \eqref{sys:pgf-multi-species} takes here the form, for $\bu=(u_{1}^\ell,u_{2}^\ell,u_{1}^i,u_{2}^i)$,
\begin{equation}\label{eq:pgf-2p1v}
    \mathbf{F}(\bu)=(f_{1}^\ell(\bu),f_{2}^\ell(\bu),f_{1}^i(\bu),f_{2}^i(\bu)),
\end{equation}
where 
\begin{subequations}\label{sys:pgf-2p1v}
\begin{align}
    f_{1}^\ell(\bu) &= \frac{\varepsilon_{1} u_{1}^i+d_1}{\varepsilon_{1}+d_1} \label{syp:pgf-2p1v_fl1}\\ 
    f_{2}^\ell(\bu) &= \frac{\varepsilon_{2} u_{2}^i+d_2}{\varepsilon_{2}+d_2} \label{syp:pgf-2p1v_fl2}\\
    f_{1}^i(\bu) &= \frac{(\beta_{11} S_1^0 u_{1}^\ell+\beta_{21} S_2^0 u_{2}^\ell) u_{1}^i+\gamma_{1}+d_1}{\Lambda_{1}} \label{syp:pgf-2p1v_fi1}\\
    f_{2}^l(\bu) &= \frac{(\beta_{12} S_1^0 u_{1}^\ell +\beta_{22} S_2^0 u_{2}^\ell) u_{2}^i+\gamma_{2}+d_2}{\Lambda_{2}},  \label{syp:pgf-2p1v_fi2}
\end{align}
\end{subequations}
with 
\[
\Lambda_{p}=\sum_{q=1}^{P}\beta_{qp} S_{q}^0+\gamma_{p} +d_p,\quad p=1,2.
\]

Solving the equation $\bF(\bz)=\bz$ in the present case involves finding $(z_1,z_2,z_3,z_4)$ such that
\begin{subequations}\label{eq:fixed-pt-2p1v}
\begin{align}
& z_1=\frac{\varepsilon_{1} z_3+d_1}{\varepsilon_{1}+d_1}, \;  z_2=\frac{\varepsilon_{2} z_4+d_2}{\varepsilon_{2}+d_2}\\
    &(\beta_{11} S_1^0z_1 +\beta_{21} S_2^0 z_2)z_3 +\gamma_{1} +d_1=(\beta_{11} S_1^0 +\beta_{21} S_2^0 +\gamma_{1}+d_1)z_3\\
     &(\beta_{12} S_1^0 z_1 + \beta_{22} S_2^0 z_2)z_4 +\gamma_{2}+d_2=(\beta_{12} S_1^0 +\beta_{22} S_2^0 +\gamma_{2}+d_2)z_4.
\end{align}
\end{subequations}
This is easily done numerically in applications and is known  by Theorem~\ref{th:exists-q-multi-pv} to have a unique solution in $[\b0,\mathbf{1})$ when $\R_0^{\eqref{sys:ode-2p1v}}>1$.

\subsection{Infectious Hematopoietic Necrosis ($P_1\leftrightarrow P_2$)}
\label{subsec:IHN}
Initially observed at fish hatcheries in Oregon and Washington in the 1950s \cite{rucker1953contagious}, Infectious Hematopoietic Necrosis (IHN) is a viral disease that affects various species of salmonids, including Sockeye Salmon (\emph{Oncorhynchus nerka}) and Chum Salmon (\emph{Oncorhynchus keta}). 
The causative agent of IHN is the Infectious Hematopoietic Necrosis virus (IHNV), which belongs to the \emph{Rhabdoviridae} family.
This virus primarily targets the hematopoietic tissues, leading to severe anaemia and necrosis. 
The incubation period ranges from 5 to 45 days \cite{spickler2007infectious}.
Clinical signs of IHN include lethargy, darkening of skin colour, haemorrhages in various organs and eventual death. 
Infected fish may display reduced swimming ability and impaired feeding behaviour due to anaemia caused by red blood cell destruction \cite{yong2019infectious}. 
IHN can have a significant economic impact on fish farms that raise young rainbow trout or salmon, with mortality rates reaching 90-95\% in highly susceptible fish species \cite{dixon2016epidemiological,spickler2007infectious}.
The virus can be transmitted horizontally through direct contact or vertically from infected parents to their offspring. Waterborne transmission is also possible, particularly in crowded aquaculture settings. 

Assume Chum Salmon is species 1 and Sockeye Salmon is species 2. 
IHN can spread within and between these two species.
From \eqref{eq:fixed-pt-2p1v}, the fixed point is solution to
\begin{subequations}\label{eq:fixed-pt-2p1v_general}
\begin{align}
& z_1=\frac{\varepsilon_{1} z_3+d_1}{\varepsilon_{1}+d_1}, \,  z_2=\frac{\varepsilon_{2} z_4+d_2}{\varepsilon_{2}+d_2}\\
    &\mathcal{R}_{01}z_3^2-\left(\mathcal{R}_{01}+\frac{ \beta_{11}S_2^0}{\gamma_{1}+d_1}(1-z_4)+1\right)z_3+ 1=0\\
     &\mathcal{R}_{02}z_4^2-\left(\mathcal{R}_{02}+\frac{ \beta_{22}S_1^0}{\gamma_{2}+d_2}(1-z_3)+1\right)z_4+ 1=0.
\end{align}
\end{subequations}
Computing exact expressions of $z_3$ and $z_4$ is not easy, but from Theorem~\ref{th:exists-q-multi-pv}, this fixed point exists.
The probabilities of IHN extinction and INH outbreak are:
\begin{align}
  \mathbb{P}_{\text{ext}}^{\text{IHN}} &=  \left\{
    \begin{aligned}
   &  \left(\frac{\varepsilon_{1} z_3+d_1}{\varepsilon_{1}+d_1}\right)^{\ell_{10}} \left(\frac{\varepsilon_{2} z_4+d_2}{\varepsilon_{2}+d_2}\right)^{\ell_{20}} z_3^{i_{10}}z_4^{i_{20}},  
   &\mathcal{R}_0^{\eqref{sys:ode-2p1v}}>1\\
    & 1, &\mathcal{R}_0^{\eqref{sys:ode-2p1v}}<1,  
    \end{aligned}
    \right. \nonumber \\
    \mathbb{P}_{\text{outbreak}}^{\text{IHN}} &=  1-\mathbb{P}_{\text{ext}}^{\text{IHN}} \label{eq:pro-inf12}
\end{align}

We perform a sensitivity analysis of the probability of an outbreak of the INH virus, assessing the impact of each parameter.
Note that here and throughout the computational work, we assume reproduction numbers larger than 1 unless otherwise specified.

Chum salmon has a lifespan of 3 to 6 years, during which females lay between 2,000 and 4,000 eggs \cite{noaa_chum_salmon}. Similarly, sockeye salmon has a lifespan of 4 to 5 years, with females laying between 2,000 and 4,500 eggs \cite{noaa_sockeye_salmon}. 
Assuming an 80\% hatch rate, the birth rates for Chum (species 1) and Sockeye (species 2) in \eqref{sys:ode-2p1v} are within the range of [355.55, 711.11] and [355.55, 800] per year, respectively. 
In a study \cite{foott2006infectious}, the incidence of infection after release ranged from 0\% to 20\%. 
These percentages do not indicate transmission rates. However, when considering the transmission rates of \eqref{sys:ode-2p1v}, we compute their values using the corresponding basic reproduction number.
Refer to Table~\ref{tab:param-and-range-2p1v} for parameter values ranges used in the sensitivity analysis. Moreover,
with an expression of the basic reproduction number of species 1 given by \(\mathcal{R}_{01}\), the transmission rates are given by
\begin{subequations}
    \begin{align}
        \beta_{11} &= \frac{\mathcal{R}_{01}(\varepsilon_{1}+d_1)(\gamma_{1}+d_1)}{\varepsilon_{1}S_1^0} \label{beta11}
        \intertext{and}
        \beta_{22} &= \frac{\mathcal{R}_{02}(\varepsilon_{2}+d_2)(\gamma_{2}+d_2)}{\varepsilon_{2}S_2^0}.\label{beta22}
    \end{align}
\end{subequations}
We assume that the infection rate of species 2 by species 1 and species 1 by species 2 is five times higher than the infection rate within species 1 and species 2, respectively, i.e., \(\beta_{21} = 5 \beta_{11}\) and \(\beta_{12} = 5 \beta_{22}\).

\begin{table}[htbp]
    \centering
    \begin{tabular}{clccc}
    \toprule
     & \textbf{Meaning} & \textbf{Range}/day &  \\
    \midrule
    $b_1$ & Birth rate species 1 & $[1, 5]$ &   \cite{noaa_chum_salmon}\\
    $b_2$ & Birth rate species 2 & $ [1, 5]$ & \cite{noaa_sockeye_salmon}\\
    $\varepsilon_{1}$ & Incubation rate of 1 & $[ 0.02, 0.2] $  &\cite{spickler2007infectious}\\
    $\varepsilon_{2}$ & Incubation rate of 2 & $[ 0.02, 0.2]$  &\cite{spickler2007infectious}\\
    $\gamma_{1}$ & Recovery rate of 1 & $[0.1, 0.33]$ &\cite{LAPATRA2000445} \\
    $\gamma_{2}$ & Recovery rate of 2 & $[0.1, 0.33]$ &\cite{LAPATRA2000445} \\
    $d_1$ & Mortality rate of 1 & [$\frac{1}{6}$, $\frac{1}{2}$]$\times$ $\frac{1}{365}$] &\\
    $d_2$ & Mortality rate of 2 & [$\frac{1}{5}$, $\frac{1}{2}$]$\times$ $\frac{1}{365}$]   & \\
    \bottomrule
    \end{tabular}
    \caption{Parameter ranges and values for IHN transmission between Chum Salmon (species 1) and Sockeye Salmon (species 2).}
    \label{tab:param-and-range-2p1v}
\end{table}

Figure~\ref{fig:PRCC_P} presents a sensitivity analysis of the probability of disease outbreak, illustrating the significant impact of incubation rates on the probability of an IHN outbreak. 
Also important are demographic parameters $b_i$ and $d_i$.
\begin{figure}[H]
    \centering
    \includegraphics[width=0.8\linewidth]{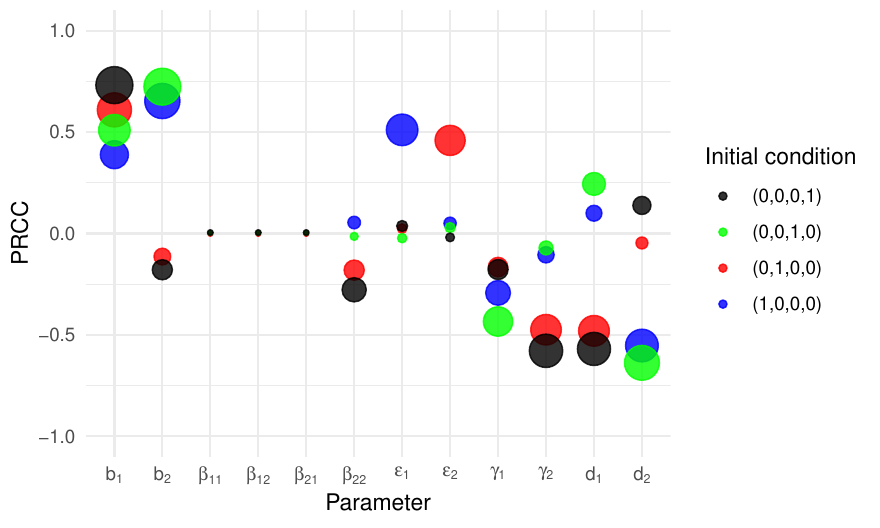}
    \caption{Partial rank correlation coefficient (PRCC) of the probability \eqref{eq:pro-inf12} of IHN outbreak for different initial conditions $y_0=(l_{10},l_{20}, i_{10}, i_{20}) \in (e_1,e_2,e_3,e_4)$. The range of parameter values remaining in Table~\ref{tab:param-and-range-2p1v}.}
    \label{fig:PRCC_P}
\end{figure}

\subsection{Transmission from wild to farmed fish ($P_1\to P_2,P_2\not\to P_1$)}
\label{subsec:wild-to-farm}
Interactions between wild and farmed fish populations are complex and crucial aspects of ecological and economic landscapes. 
In the scenario where wild fish could penetrate the space where fish are farmed, the potential transmission of diseases between these populations becomes a significant concern \cite{arechavala2013reared}. 
Note that contamination of wild fish by farm fish is also of concern \cite{johansen2011disease,krkosek2007declining}.

The situation we consider here has wild fish (species 1) able to introduce diseases to farmed populations, while farmed fish (species 2) cannot transmit diseases to wild fish, because they are raised in captivity. 
This asymmetry in transmission dynamics implies that \eqref{sys:CTMC_2p1v} is considered with $\beta_{21}\neq 0$, while $\beta_{12}=0$.

It follows that the fixed point relations \eqref{eq:fixed-pt-2p1v} take the form
\begin{subequations}\label{eq:fixed-pt-2p1v-2noinf1}
    \begin{align}
        & z_1=\frac{\varepsilon_{1} z_3+d_1}{\varepsilon_{1}+d_1}, 
        \quad z_2=\frac{\varepsilon_{2} z_4+d_2}{\varepsilon_{2}+d_2},
        \quad z_3=\frac{1}{\mathcal{R}_{01}},
        \label{eq:fixed-pt-2p1v-2noinf1-q1q2q3}\\
         & \mathcal{R}_{02}z_4^2-\left(\mathcal{R}_{02}+c+1\right)z_4+ 1=0. \label{eq:fixed-pt-2p1v-2noinf1-q4}
    \end{align}
\end{subequations}
where we have denoted $c=\beta_{21}(\mathcal{R}_{01}-1)(\gamma_{1}+d_1)/(\beta_{11}(\gamma_{2}+d_2))$ and 
\begin{equation}\label{eq:R02}
    \mathcal{R}_{02}=\frac{\beta_{22}\varepsilon_{2}S_2^0}{(\varepsilon_{2}+d_1)(\gamma_{2}+d_1)}      
\end{equation}
the basic reproduction for the pathogen in species 2 \emph{in the absence of contact} with species 1.

To solve \eqref{eq:fixed-pt-2p1v-2noinf1-q4}, we first consider the discriminant and using the fact the $c>0$, 
\[
D = \left(\mathcal{R}_{02} + c + 1\right)^2 - 4\mathcal{R}_{02}> \left(\mathcal{R}_{02} + 1\right)^2 - 4\mathcal{R}_{02}=(\mathcal{R}_{02}-1)^2 > 0.
\]
Therefore, \eqref{eq:fixed-pt-2p1v-2noinf1-q4} has two real solutions $z_4^-$ and $z_4^+$ given by
\[ 
z_4^- = \frac{\left(\mathcal{R}_{02} + c + 1\right) - \sqrt{D}}{2\mathcal{R}_{02}}
\quad\text{and}\quad
z_4^+ = \frac{\left(\mathcal{R}_{02} + c + 1\right) + \sqrt{D}}{2\mathcal{R}_{02}}.
\]
Given that $\mathcal{R}_{02} > 1$, both $z_4^-$ and $z_4^+$ are positive, $z_4^+ > z_4^-$, and $z_4^+z_4^- = 1/\mathcal{R}_{02}$.
Moreover, since $ D > (\mathcal{R}_{02}-1)^2 $, we obtain
\begin{equation*}
    z_4^+ > 1+ \frac{c}{2\mathcal{R}_{02}}  \Longrightarrow  z_4^+z_4^->  z_4^- \Longrightarrow   z_4^- <\frac{1}{\mathcal{R}_{02}}.
\end{equation*}
Thus, $z_4^{+}> 1$, $z_4^{-} <1$ and solution to the fixed point problem has
\begin{equation}\label{eq:exp-q4-2p1v-no1}
z_4:=z_4^-= \frac{\left(\mathcal{R}_{02} +c+1\right) - \sqrt{\left(\mathcal{R}_{02}+c+1\right)^2 - 4\mathcal{R}_{02}}}{2}.
\end{equation}
Then the probability of extinction of the disease in a wild-to-farmed ($W\to F$) context
\begin{equation*}
  \mathbb{P}_{\text{ext}}^{W\to F}=  \left\{
    \begin{aligned}
   & z_1^{\ell_{10}}z_2^{\ell_{20}}z_3^{i_{10}}z_4^{i_{20}}  &\mathcal{R}_{01}, \mathcal{R}_{02}>1\\
    & 1, &\mathcal{R}_{01},\mathcal{R}_{02}\leq 1,
    \end{aligned}
    \right. 
\end{equation*}
where $z_1,z_2$ and $z_3$ are given by \eqref{eq:fixed-pt-2p1v-2noinf1-q1q2q3} and $z_4$ by \eqref{eq:exp-q4-2p1v-no1}; furthermore,
\begin{equation}\label{eq:P_2nonfec1}
    \mathbb{P}_{\text{outbreak}}^{W\to F}=  1- \mathbb{P}_{\text{ext}}^{W\to F}.
\end{equation}

For the sensitivity analysis of \eqref{eq:P_2nonfec1}, we utilized the assumed parameter values provided in Table~\ref{tab:param_ranges_2noninfec1}. The parameter ranges for $\beta_{11}$ and $\beta_{21}$ were computed using the relations given in \eqref{beta11} and \eqref{beta22}.

\begin{table}[H]
    \centering
    \begin{tabular}{lcc}
        \toprule
        \textbf{Parameter} & \textbf{Range} & \textbf{Value} \\
        \midrule
        $\mathcal{R}_0$ & $[0.1, 5]$ & - \\
        $b_1$ & $[2, 20]$ & 8 \\
        $b_2$ & $[2, 20]$ & 10 \\
        $\varepsilon_{1}$ & $[0.01, 0.2]$ & 0.05 \\
        $\varepsilon_{2}$ & $[0.01, 0.2]$ & 0.05 \\
        $\gamma_{1}$ & $[0.01, 0.1]$ & 0.02 \\
        $\gamma_{2}$ & $[0.01, 0.1]$ & 0.02 \\
        $d_1$ & [$\frac{1}{10}$, $\frac{1}{2}$]$\times$ $\frac{1}{365}$ & -\\
        $d_2$ & [$\frac{1}{10}$, $\frac{1}{2}$]$\times$ $\frac{1}{365}$& - \\
        \bottomrule
    \end{tabular}
    \caption{Parameter ranges and values for transmission from wild to farmed fish.}
    \label{tab:param_ranges_2noninfec1}
\end{table}

Sensitivity analysis of disease outbreak probability $\mathbb{P}_{\text{outbreak}}^{W\to F}$ \eqref{eq:P_2nonfec1} (Figure~\ref{fig:PRCC_P_2noinfec1}) shows that birth ($b_1$), death ($d_1$) and recovery ($\gamma_{1}$) rates of species 1 are key drivers of the probability of an outbreak, regardless of initial conditions. 
Interestingly, for the incubation rates ($\varepsilon_{1}$ and $\varepsilon_{2}$) of species 1 and 2, respectively, there is a notable difference in impact based on the initial condition. 
The PRCC is 0.3 and 0.22 for the initial condition starting with one latent individual in Species 1 and Species 2, respectively, while it is very small when initially starting with infected individuals.

In contrast, factors such as infection rates $\beta$ of both species and death rate of species 2 exert minimal influence on the probability of an outbreak.
Moreover, specific cases involving parameters $\gamma_{2}$ and $d_1$ exhibit interesting patterns depending on where the infection starts. 
For instance, these parameters show a positive impact when the infection originates within species 1 but have a negative impact if it starts in species 2.

\begin{figure}[H]
    \centering
    \includegraphics[width=0.8\linewidth]{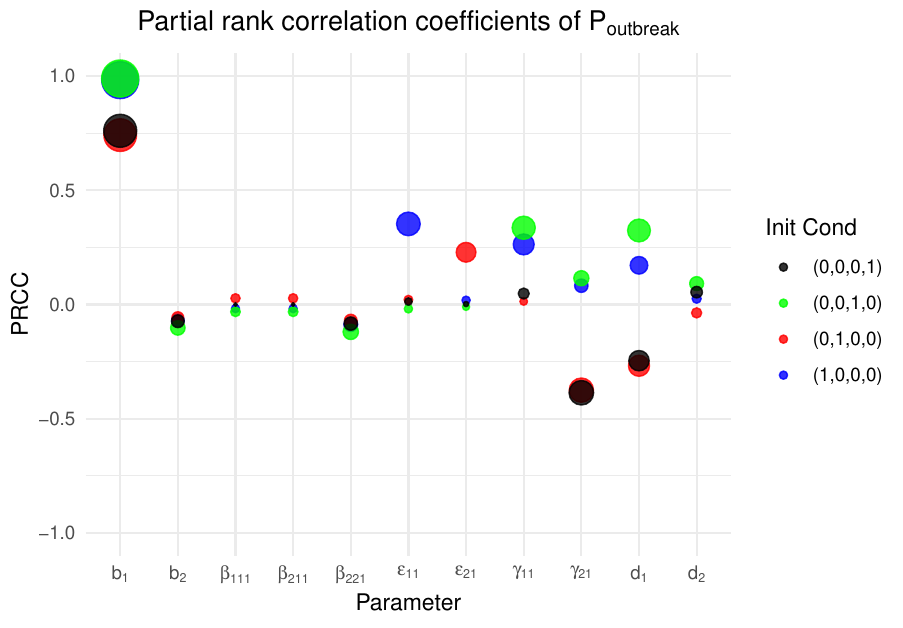}
    \caption{PRCC of the probability \eqref{eq:P_2nonfec1} of disease outbreak $\mathbb{P}_{\text{outbreak}}^{W\to F}$ for four different initial conditions $y_0=(\ell_{10},\ell_{20}, i_{10}, i_{20}) \in (e_1,e_2,e_3,e_4)$. 
    Parameter values ranges in Table~\ref{tab:param_ranges_2noninfec1}.}
    \label{fig:PRCC_P_2noinfec1}
\end{figure}

Focusing on sensitivity to mortality rates, we observe that effects can be both positive and negative.
To investigate this further, we consider in Figure~\ref{fig:d_1_d_2_2noinf1} the probability $\mathbb{P}_{\text{outbreak}}^{W\to F}$  of an outbreak as a function of the mortality rates $d_1$ and $d_2$ of species 1 and 2, respectively. 
When the infections start with one infectious individual in species 1, the probability of an outbreak is between 0.55 and 1. 
When the disease starts with one infectious individual of species 2, this range narrows to between 0.9 and 1.
For initial infections in species 1, it is primarily the mortality rate of species 2 that impacts the value of this probability. 
In contrast, when the infection is initiated by species 2, both mortality rates play a role in determining the outcome. 
Specifically, as the mortality rate $d_1$ of species 1 increases, there is a decrease in the probability of an outbreak, while an increase in $d_2$ leads to a higher probability value. 
Thus, initial infections in different species can lead to varying probabilities of outbreak.

\begin{figure}[H]
    \centering
    \includegraphics[width=\textwidth]{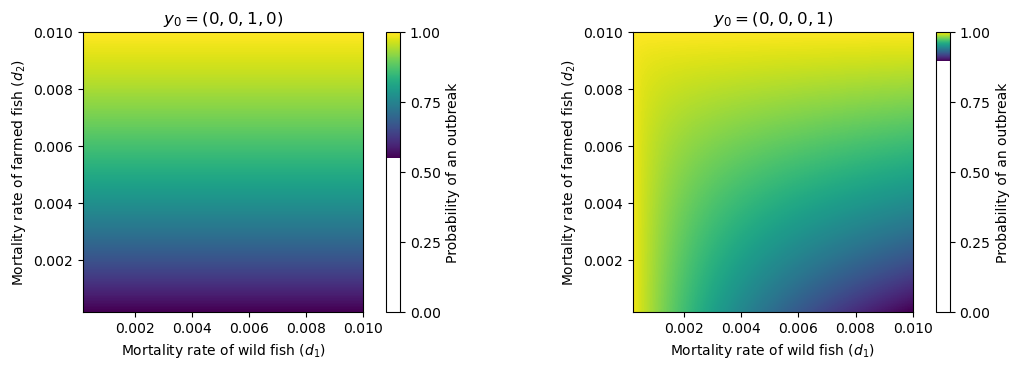}
    \caption{Probability \eqref{eq:P_2nonfec1} of disease outbreak as a function of mortality rates of wild and farmed fish.  
   For two initial conditions $y_0=(\ell_{10},\ell_{20}, i_{10},i_{20})\in (e_3,e_4)$. 
   Parameter values in Table~\ref{tab:param-and-range-2p1v}, with $\beta_{11}=\beta_{22}=10^{-5}$, $\beta_{21}=5\beta_{11}$.}
    \label{fig:d_1_d_2_2noinf1}
\end{figure}

\subsection{Viral Hemorrhagic Septicaemia ($P_1\to P_2,P_2\not\to$)}
\label{subsec:VHS}
Viral Hemorrhagic Septicaemia (VHS) is a highly contagious and deadly disease that affects various species of fish in both cultured and wild populations. It is prevalent in freshwater and marine environments across several regions of the Northern Hemisphere \cite{gagne2007,meyers1995,takano2001}. VHS is caused by the Viral Hemorrhagic Septicemia Virus (VHSV) and was first isolated in Alaska from skin lesion material of two Pacific cod \emph{Gadus macrocephalus} \cite{meyers1992identification}. 
VHSV is known for causing hemorrhagic septicemia, which leads to severe internal bleeding and organ damage in infected fish. The symptoms vary depending on the species and the stage of infection. Some common signs include lethargy, loss of appetite, abnormal swimming behaviour and external haemorrhaging.
Over 100 species of freshwater and marine fish have been reported to be naturally or experimentally susceptible to VHSV \cite{batts2020viral}. These species include the Pacific herring \emph{Clupea pallasii} \cite{kocan2001epidemiology,meyers1994association}.

We consider the latter as species 1 in \eqref{sys:ode-2p1v}. 
We suppose that the latent stage is an enzootic stage, the infectious stage is the combination of disease amplification and outbreak stage and the recovery stage combines recovery and refractory as suggested by \cite{mitro2008viral}. 
Now, while VHSV can infect a wide range of fish, not all infected individuals show signs of clinical disease and not all are capable of transmitting the virus to others \cite{ord1976viral}.
So we consider as species 2 the hybrid fry, Steelhead trout (\textit{Salmo gairdneri})$\times$Coho Salmon (\textit{Oncorhynchus kisutch}), which can get infected, but whose ability to transmit VHS has not been demonstrated \cite{ord1976viral}, further simplifying the situation by making the assumption that this hybrid fry cannot transmit VHS. As a consequence, \(\beta_{121} = 0\) and \(\beta_{221}=0\). Then, the fixed point relations \eqref{eq:fixed-pt-2p1v} become
\begin{subequations}\label{eq:fixed-pt-2p1v-2noinf}
        \begin{align}
       & z_1=\frac{\varepsilon_{1} z_3+d_1}{\varepsilon_{1}+d_1}, \; z_2=\frac{\varepsilon_{2} z_4+d_2}{\varepsilon_{2}+d_2}\\
    &\mathcal{R}_{01} z_3^2-(\mathcal{R}_{01}+1) z_3+1=0\\
     &\beta_{21} S_1^0 z_1 z_4 +\gamma_{2}+d_2=(\beta_{21} S_1^0+\gamma_{2}+d_2)z_4
\end{align}
\end{subequations}
where we have denoted
\begin{equation}\label{eq:R01}
 \mathcal{R}_{01}=\frac{\beta_{11}  \varepsilon_{11} S_1^0}{(\varepsilon_{11}+d_1)(\gamma_{11}+d_1)}
\end{equation}
the basic reproduction for species 1 and the virus \emph{in the absence of contact} with species 2. 

Solving \eqref{eq:fixed-pt-2p1v-2noinf} yields
\[(1,1,1,1) \; \text{and} \; \left(z_1,z_2,z_3,z_4\right)\]
 where 
 \begin{equation}\label{eq:exp-q-2p1v}
    z_1=\frac{\varepsilon_{1} z_3+d_1}{\varepsilon_{1}+d_1}, \;z_2=\frac{\varepsilon_{2} z_4+d_2}{\varepsilon_{2}+d_2}, \;z_3=\frac{1}{\mathcal{R}_{01}} , \; z_4=\frac{1}{c +1}
 \end{equation}
with $c=\beta_{21}(\mathcal{R}_{01}-1)(\gamma_{1}+d_1)/(\beta_{11}(\gamma_{2}+d_2))$.

Therefore, the probability of extinction (no outbreak) is one if $\mathcal{R}_{01} \leq 1$, but less than one if $\mathcal{R}_{01}>1$. Given the initial conditions $L_1(0)=\ell_{10}$, $I_1(0)=i_{10}$, $L_2(0)=\ell_{20}$, and $I_2(0)=i_{20}$, it follows from the independent branching process approximation that the probabilities of the disease extinction and the disease outbreak are

\begin{align}
    \mathbb{P}_{\text{ext}}^{\text{VHS}} &= \left\{
    \begin{aligned}
    & z_1^{\ell_{10}}z_2^{\ell_{20}}z_3^{i_{10}}z_4^{i_{20}},  
    & \mathcal{R}_{01}>1 \\
    & 1, &\mathcal{R}_{01}\leq 1,
    \end{aligned}
    \right. \nonumber\\
    \mathbb{P}_{\text{outbreak}}^{\text{VHS}} &= 1- \mathbb{P}_{\text{ext}}^{\text{VHS}},
    \label{eq:P_2noinfec}
\end{align}
where the expressions of $z_1,z_2,z_3$ and $z_4$ are given in \eqref{eq:exp-q-2p1v}.

On average, a female Pacific herring lays 20,000 eggs each year \cite{adfg_herring}.
If we suppose that the viable percentage of the population is around 60\%, then $b_1=12000/year\approx 33/day$. 
In laboratory experiments with infected herring, it was observed that shed VHSV can be identified in water within 4-5 days after exposure (PE), preceding the onset of host mortality due to the disease. The peak of viral shedding occurs between 6-10 days PE \cite{garver2021characterization}. 
We then consider the incubation period of 5-100 days.  
The duration of recovery varied depending on the phase of the epizootic. During the acute phase, which occurred around day 13 post-exposure, virus loads in tissues were significantly higher compared to the recovery phase, which spanned days 30 to 42 \cite{Hershberger2010}; for sensitivity analysis of this duration is taken between 20 and 100 days. 
We suppose hybrid fry averaging 12,000 eggs per year with 80\% viability, giving $b_2=9600/year=26.3/day$. 
For the sensitivity analysis of the probability $\mathbb{P}_{\text{outbreak}}^{\text{VHS}}$ of VHS outbreak, the range of parameter values are given in Table~\ref{tab:param-and-range-2p1v-2noinf}, \eqref{beta11} and \eqref{beta22} is used to compute infection rates. 
 
\begin{table}[H]
    \centering
    \begin{tabular}{clccc}
    \toprule
     & \textbf{Meaning} & \textbf{Range} & \textbf{Values} & Ref  \\
    \midrule
     $\mathcal{R}_{01}$ & Basic reproduction number & $[0.1, 5]$ &  &  \\
    $b_1$ & Birth rate Pacific herring & $[10, 50]$ & 33 &  \cite{adfg_herring} \\
    $b_2$ & Birth rate hybrid fry & $ [10, 30]$ & 26.3 & \\
    $\varepsilon_{1}$ & Incubation rate of Pacific herring & $[ 0.01, 0.2] $ & 0.02 &\cite{garver2021characterization}\\
    $\varepsilon_{2}$ & Incubation rate of hybrid fry  & $[ 0.01, 0.2]$ & 0.02 &\cite{garver2021characterization}\\
    $\gamma_{1}$ & Recovery rate of Pacific herring & $[0.01, 0.05]$ & 0.03 &\cite{Hershberger2010} \\
    $\gamma_{2}$ & Recovery rate of hybrid fry  & $[0.0
    1, 0.05]$ & 0.03 &\cite{Hershberger2010} \\
    $d_1$ & Mortality rate of species 1 & [$\frac{1}{15}$, $\frac{1}{2}$]$\times$ $\frac{1}{365}$ & 3.$10^{-4}$ & \cite{adfg_herring,stokesbury2002natural}\\
    $d_2$ & Mortality rate of species 2 &[$\frac{1}{15}$, $\frac{1}{2}$]$\times$ $\frac{1}{365}$ &  3.$10^{-4}$  & \\
    \bottomrule
    \end{tabular}
    \caption{Parameter range and values for \eqref{sys:ode-2p1v} with two species, Pacific herring (species 1) and hybrid fry (species 2) and one pathogen, VHS.}
    \label{tab:param-and-range-2p1v-2noinf}
\end{table}

In Figure~\ref{fig:PRCC_P_2noinfec}, a sensitivity analysis of the probability of disease outbreak $\mathbb{P}_{\text{outbreak}}^{\text{VHS}}$ is presented for four different initial conditions. 
The birth rate $b_1$ of species 1 stands out as having a significant impact on the probability of an outbreak regardless of the initial condition. 
When the infection begins with one latent and one infectious individual in species 2, parameters $b_1$, $\gamma_{1}$, $\gamma_{2}$, and $d_1$ all exhibit a similar level of impact on the probability of an outbreak.
The mortality rate $d_1$ of species 1 shows varying effects on the probability of an outbreak depending on where the disease originates. 
When the infection starts in species 2, $d_1$ has a significantly positive impact on the likelihood of an outbreak. However, this impact becomes negative when the disease originates in species 1.

Amongst all parameters analysed, certain factors show no discernible impact on the probability of an outbreak. This includes parameter $b_2$, as well as transmission rates $\beta_{11}$ and $\beta_{21}$. 

These findings highlight which variables have little to no effect on influencing whether a disease outbreak will occur. Moreover, 
 the influence of different initial conditions on outcomes is a crucial aspect to consider.

\begin{figure}[H]
    \centering
    \includegraphics[width=0.8\textwidth]{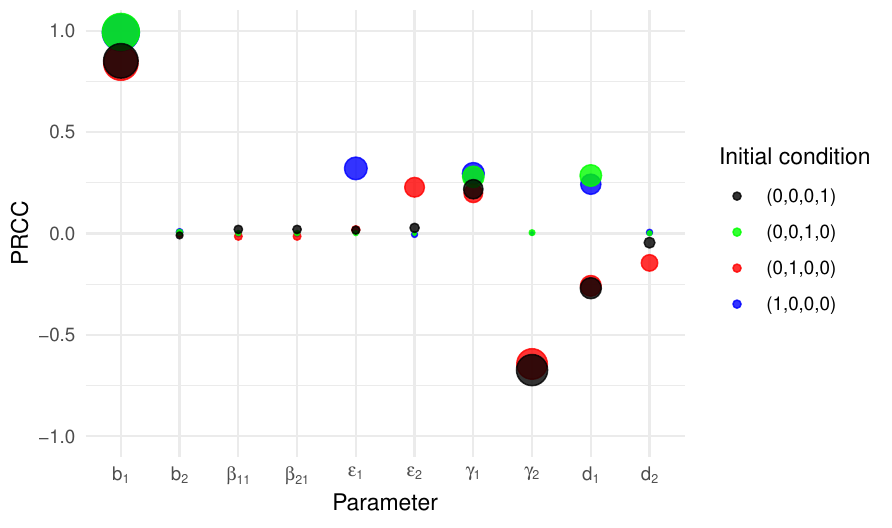}
    \caption{PRCC of the probability of disease outbreak \eqref{eq:P_2noinfec} for different initial conditions $y_0=(\ell_{10},\ell_{20}, i_{10}, i_{20}) \in (e_1,e_2,e_3,e_4)$. Parameter values ranges in Table~\ref{tab:param-and-range-2p1v-2noinf}. }
    \label{fig:PRCC_P_2noinfec}
\end{figure}

\subsection{The pathogen is endemic in one species ($P_1^\star\to P_2$)}
\label{subsec:intro-from-endemic}
The problem motivating this study concerns vagrant species coming in contact with resident species while bearing a pathogen the resident species has not been exposed to yet.

To model this situation, we artificially impose that species 1 be at an endemic equilibrium while species 2 be at a disease-free equilibrium.
A situation with mixed equilibria of this type is not possible in \eqref{sys:ode-2p1v}, so we ``cheat'': we start by assuming that both species are not in contact and that species 1 is at the endemic equilibrium $E_1^\star:=(S_1^\star,L_1^\star,I_1^\star,R_1^\star)$, in which
\begin{equation}
\label{eq:EEP_E1}
    E_1^\star = 
    \left(
    \frac{S_1^0}{\mathcal{R}_{01}},\; 
    \frac{(\mathcal{R}_{01}-1)(\gamma_{1}+d_1)d_1}{\beta_{11} \varepsilon_{1}},\;
    \frac{(\mathcal{R}_{01}-1)d_1}{\beta_{11}},\;
    \frac{(\mathcal{R}_{01}-1)\gamma_{1}}{\beta_{11}}
    \right),
\end{equation}
where $\mathcal{R}_{01}$ given by \eqref{eq:R01} is the basic reproduction for the pathogen in species 1 in the absence of contact with species 2.
We suppose that this determines the dynamics of the pathogen in species 1, which we now ignore except insofar as the infecting potential of the $I_1^\star$ individuals from species 1 infectious with the pathogen potentially coming into contact with susceptible individuals from species 2.
We then consider the second population, focusing on conditions leading to the disease becoming established there.

\subsubsection{The ODE introduction model}
\label{subsec:ODE-introduction-model}

Given a prevalence of infection $I_1^\star$ in species 1, the infection dynamics in species 2 is governed by 
\begin{subequations}\label{sys:ode-p2}
\begin{align}
\dot  S_2 &=b_2-\beta_{21}I_1^\star S_2-\beta_{22} S_2 I_2-d_2S_2\\
\dot  L_2 &=\beta_{21} I_1^\star S_2+\beta_{22} S_2 I_2-(\varepsilon_{2}+d_2) L_2 \label{sys:ode-p2_dL2}\\
\dot I_2 &=\varepsilon_{2} L_2-(\gamma_{2} +d_2) I_2 \\
\dot R_2& =\gamma_{2} I_2-d_2 R_2,
\end{align}
\end{subequations}
considered with nonnegative initial conditions.

Model \eqref{sys:ode-p2} is not a classic immigration model, since the term $\beta_{21}I_1^\star$ is factor of $S_2$, not a constant. 
However, similarly to immigration models, the term $\beta_{21}I_1^\star S_2$ in \eqref{sys:ode-p2_dL2} precludes the existence of a disease-free equilibrium for \eqref{sys:ode-p2}.
As a consequence, no basic reproduction number can be computed for \eqref{sys:ode-p2}. The method of \cite{almarashi2019effect} is not applicable here since immigration occurs at the \emph{per capita} rate $\beta_{21}I_1^\star$, not at a constant rate.
The term $\beta_{21}I_1^\star$ cannot be considered as encoding horizontal transmission either, at least not in the usual sense, since inflow into \eqref{sys:ode-p2_dL2} is function of $S_2$, not $I_2$.

As a consequence, a thorough analysis of properties of \eqref{sys:ode-p2} is difficult and beyond the scope of this work.
Instead, we focus on simple properties as well as computational considerations.

System \eqref{sys:ode-p2} admits a unique equilibrium, an endemic equilibrium taking the form $E_2^\star:=(S_2^\star,L_2^\star,I_2^\star,R_2^\star)$, where 
\[ 
S_2^\star=\frac{\varepsilon_2+d_2}{\beta_{21}I_1^\star+\beta_{22}I_2^\star}L_{2}^\star,\;L_2^\star=\frac{\gamma_2+d_2}{\varepsilon_2}I_{2}^\star,\;R_2^\star=\frac{\gamma_{2}I_2^\star}{d_2}
\] 
and \(I_2^\star\) is a root of the second order polynomial 
\begin{equation}\label{eq:poly_I2}
    \beta_{22}(I_2^\star)^2-((\mathcal{R}_{02}-1)d_2-\beta_{21}I_1^\star)I_2^\star-\frac{\beta_{21}d_2I_1^\star\mathcal{R}_{02}}{\beta_{22}}=0,
\end{equation}
where $\mathcal{R}_{02}$ given in \eqref{eq:R02} is the basic reproduction for species 2 and the virus \emph{in the absence of contact} with species 1.

Since the coefficients \(\beta_{22}>0\) and $\beta_{21}\mathcal{R}_{02}I_1^\star/\beta_{22}>0$, Descartes rule of signs implies that the polynomial~\eqref{eq:poly_I2} has a unique positive solution. The expression of this solution is given by
\begin{equation}
    I_2^\star = \frac{(\mathcal{R}_{02}-1)d_2-\beta_{21}I_1^\star+
    \sqrt{((\mathcal{R}_{02}-1)d_2-\beta_{21}I_1^\star)^2+4\mathcal{R}_{02}d_2\beta_{21}I_1^\star}}{2\beta_{22}}.
\end{equation}

\begin{figure}[H]
    \centering
    \includegraphics[width=0.8\textwidth]{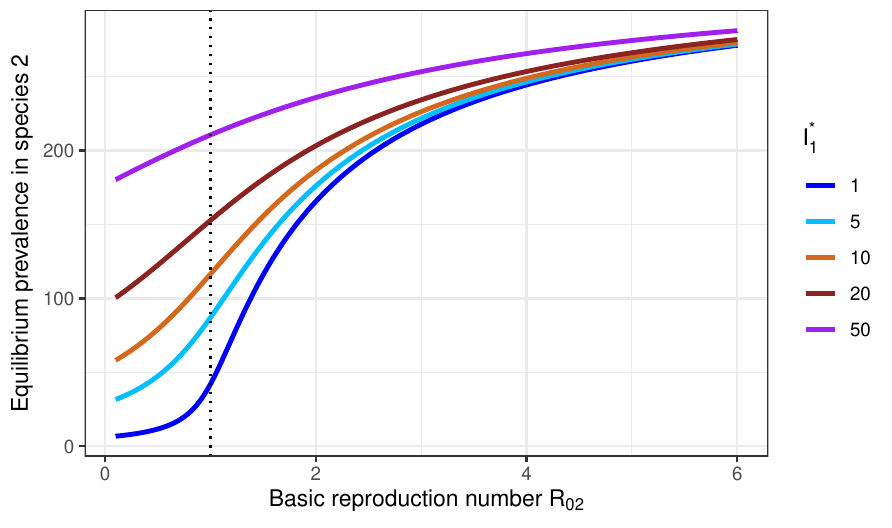}
    \caption{Equilibrium prevalence of infection $I_2^\star$ in population 2 as a function of the reproduction number $\R_{02}$ for the pathogen in species 2. The different curves correspond to different values of prevalence $I_1^\star$ in population 1.}
    \label{fig:I2star_fct_I1star_R02}
\end{figure}

In Figure~\ref{fig:I2star_fct_I1star_R02}, we show how (equilibrium) prevalence of infection in species 2 depends on prevalence of infection in species 1.
For low values of the reproduction number in species 2, the situation is quite dependent on the prevalence of infection in the ``introducing species'', but as the reproduction number increases, this dependence diminishes to the point of the situation becoming indistinguishable for large values of $\R_{02}$.

\subsubsection{The CTMC introduction model}
Note that it is not possible here to use a multitype branching process approximation, for roughly the same reasons that a basic reproduction number does not exist for the deterministic model \eqref{sys:ode-p2}.
Indeed, while branching processes incorporating \emph{immigration} exist \cite{heathcote1965branching,pakes1971branching}, they assume that immigration is a process that is independent from the population immigrants are joining.
In our model, this is not true, since ``immigration'' depends on the susceptible population $S_2$.
So we focus here on the computational analysis of the CTMC.
In all results presented here, 10,000 simulations (realisations) of the CTMC associated to \eqref{sys:ode-p2} are used for each data point.

\begin{figure}
    \centering
    \includegraphics[width=0.8\textwidth]{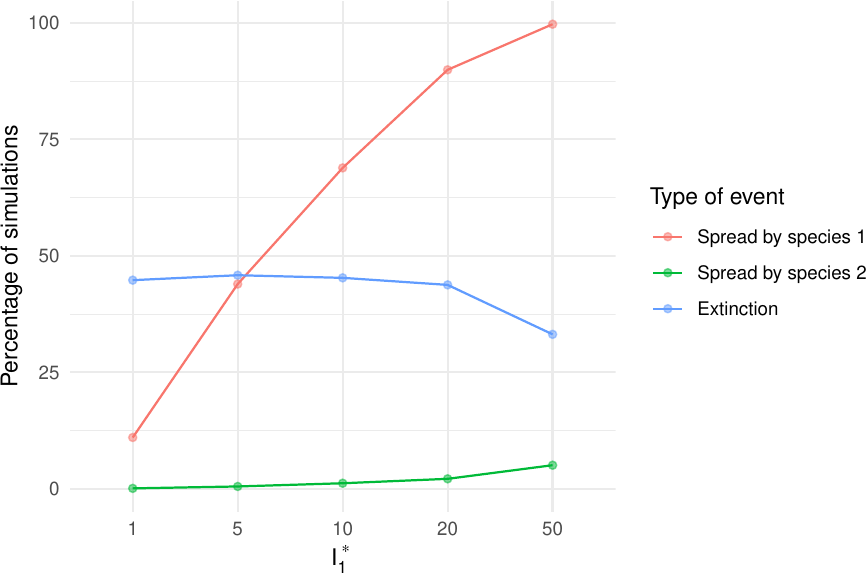}
    \caption{Percentage of 10K realisations in which spread by species 1 (introductions) and species 2 occurs. 
    Also shown is the percentage of realisations where introductions by species 1 are followed by extinctions.}
    \label{fig:pct-intro-comm-ext}
\end{figure}

Figure~\ref{fig:pct-intro-comm-ext} shows the percentage of 10K simulations in which at least one transmission occurs from species 1 to species 2 (red curve) and from species 2 to species 2 (green curve).
Simulations assume that $\R_{02}=1.5$ and are run for 90 days.
Note that the latter type of transmission always requires that introduction by species 1 has taken place.
Additionally, we show the percentage of realisations in which spread by species 1 occurs followed by extinctions, where we characterise an extinction of the infection in species 2 as a moment when the total number of infected in species 2, $L_2(t)+I_2(t)$, is zero after having been positive.

This illustrates a very important part of the introduction process. 
In keeping with the terminology in \cite{ArinoBoelleMillikenPortet2021,ArinoBajeuxPortetWatmough2020}, let us assess success of an introduction from the perspective of the pathogen.
In view of Figure~\ref{fig:pct-intro-comm-ext}, what drives successful introductions is the size of the introduction, i.e., the so-called \emph{inoculum size}.
However, before it becomes established in species 2 (the consequence of which is shown in a deterministic context by Figure~\ref{fig:I2star_fct_I1star_R02}), the infection in species 2 must ``survive'' the \emph{stochastic phase} of the epidemic; see, e.g., the one location case in \cite{ArinoMilliken2022b}.
This illustrates that spillover events are often unsuccessful, as observed in the zoonotic case with bats \cite{sanchez2022strategy}.

\begin{figure}
    \centering
    \includegraphics[width=\textwidth]{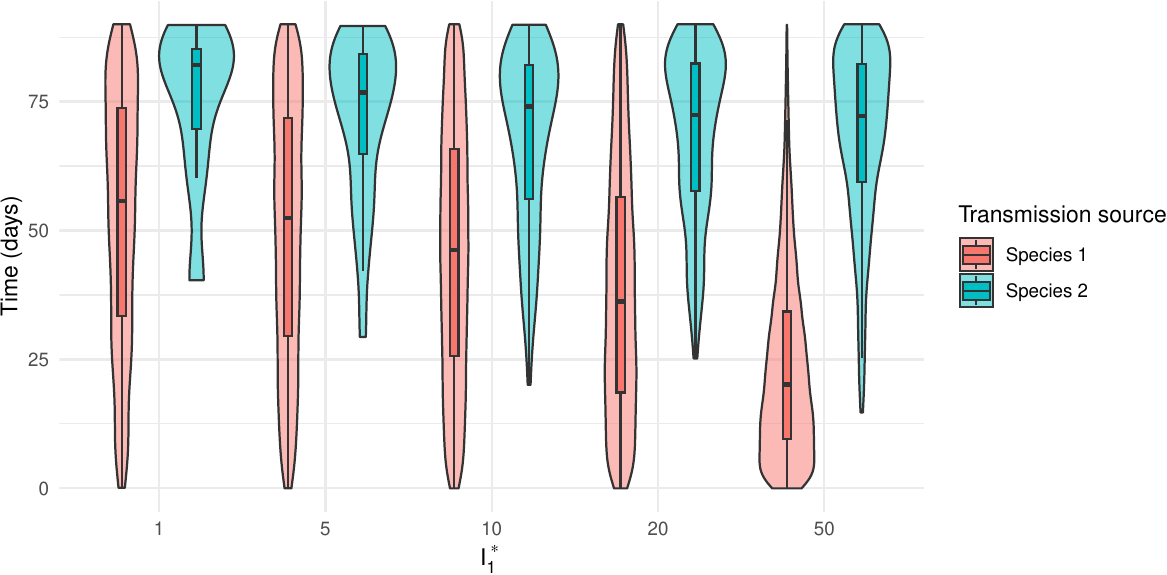}
    \caption{Distribution of the times at which the first infection event occurs in species 2 following contact with an infected individual from species 1 (red) and species 2 (blue), for different values of $I_1^\star$. Boxplots within the violins show the median, interquartile range and whiskers extents.}
    \label{fig:time-to-transmission-fct-source}
\end{figure}

To better understand this issue, consider Figure~\ref{fig:time-to-transmission-fct-source}, where we show violin plots of the distribution of times at which the first infection in species 2 arises stemming from contact with, respectively, species 1 and species 2.
These values are from the same simulations as used in Figure~\ref{fig:pct-intro-comm-ext} and thus represent the percentages shown there of 10K simulations.
First, consider infections with source species 1, i.e., introductions into species 2.
We observe that as the prevalence $I_1^\star$ in species 1 increases, the time to first introduction progressively diminishes, with the interquartile range covering smaller and smaller values.
Now consider the timing of infections originating from species 2, i.e., taking place after the infection has become somewhat established in species 2.
There, we observe that as the prevalence $I_1^\star$ increases, times to the first transmission decrease, but not as acutely as those to the first introduction.
This confirms our earlier observation that introduction and establishment, although evidently correlated, are not as tightly tied as can be expected.

\section{Case of two pathogens and two species} \label{sec:2p2v}
Here the dynamics with $P=2$ species and $V=2$ pathogens are considered.
While the cases in Section~\ref{sec:2p1v} are more tractable mathematically, the situation here is more realistic.
In practice, the collaboration motivating this work is interested in over a dozen pathogens potentially infecting four fish species.
We do not consider such a general situation here, but illustrate the computational complexities that arise even when $P=V=2$.
One particularly interesting feature is the existence of mixed equilibria, i.e., equilibria in which one of the pathogens is present and the other absent.

\subsection{Deterministic model and basic analysis}
\begin{subequations}\label{sys:ode-2p2v}
    \begin{align}
        &\dot S_1 = b_1 - \left(\sum_{q=1}^{2}\sum_{v=1}^{2}\beta_{1qv}I_{qv} + d_1 \right)S_1, 
        &\dot S_2 &= b_2 - \left(\sum_{q=1}^{2}\sum_{v=1}^{2}\beta_{2qv}I_{qv} + d_2 \right)S_2,  \\ \label{sys:ode-2p2v_S}
        &\dot L_{11} = \sum_{q=1}^{2}\beta_{1q1}I_{q1} S_1 - (\varepsilon_{11} + d_1)L_{11}, 
        &\dot L_{21}& = \sum_{q=1}^{2}\beta_{2q1}I_{q1} S_2 - (\varepsilon_{21} + d_2)L_{21},  \\ \label{sys:ode-2p2v_Lq1}
        &\dot L_{12} = \sum_{q=1}^{2}\beta_{1q2}I_{q2} S_1 - (\varepsilon_{12} + d_1)L_{12},  
        &\dot L_{22}& = \sum_{q=1}^{2}\beta_{2q2}I_{q2} S_2 - (\varepsilon_{22} + d_2)L_{22},  \\
        &\dot I_{11} = \varepsilon_{11}L_{11} - (\gamma_{11} + d_1)I_{11}, 
        &\dot I_{21}& = \varepsilon_{21}L_{21} - (\gamma_{21} + d_2)I_{21},  \\ \label{sys:ode-2p2v_Iq1}
        &\dot I_{12} = \varepsilon_{12}L_{12} - (\gamma_{12} + d_1)I_{12}, 
        &\dot I_{22}& = \varepsilon_{22}L_{22} - (\gamma_{22} + d_2)I_{22},  \\
        &\dot R_1 = \sum_{v=1}^{2}\gamma_{1v}I_{1v} - d_1R_1, 
        &\dot R_2& = \sum_{v=1}^{2}\gamma_{2v}I_{2v} - d_2R_2. \label{sys:ode-2p2v_R}
    \end{align}
\end{subequations}

The disease-free equilibrium of~\eqref{sys:ode-2p2v} is 
\begin{equation}
\bE_0^{\eqref{sys:ode-2p2v}}=(S_1^0,S_2^0,0_{\mathbb{R}^{10}}), \; \text{ with } S_1^0=\frac{b_1}{d_1} \text{ and }  S_2^0=\frac{b_2}{d_2}.
\end{equation}
To compute the basic reproduction number $\R_0$, observe that the matrices $\mathbf{G}_{12}$ and $\mathbf{W}$ derived in Section~\ref{sec:basic_analysis_ODE} take here the form
\[
\mathbf{G}_{12}=
\begin{bmatrix}
\mathbf{G}_{12}^{11}  & \mathbf{G}_{12}^{12}  \\
\mathbf{G}_{12}^{21} & \mathbf{G}_{12}^{22}  \\
\end{bmatrix},
\]
where $\mathbf{G}_{12}^{11} =S_1^0\diag(\beta_{111},\beta_{112})$, $\mathbf{G}_{12}^{12}=S_1^0\diag(\beta_{121},\beta_{122})$, $\mathbf{G}_{12}^{21}= S_2^0\diag(\beta_{211},\beta_{212})$ and $\mathbf{G}_{12}^{22}=S_2^0\diag(\beta_{221},\beta_{222})$ are diagonal matrices, and $S_1^0$ and $S_2^0$ are scalars. 
The matrix $\mathbf{W}$ is block lower triangular, with blocks
\[
\mathbf{W_{11}}=\diag(d_{1} + \varepsilon_{11}, d_{1} + \varepsilon_{12},d_{2} + \varepsilon_{21},d_{2} + \varepsilon_{22}),
\]
\[
\mathbf{W_{21}}=\diag(\varepsilon_{11},\varepsilon_{12},\varepsilon_{21},\varepsilon_{22}) 
\]
and
\[
\mathbf{W_{22}}=\diag(d_{1} + \gamma_{11} ,d_{1} + \gamma_{12},d_{2} + \gamma_{21},d_{2} + \gamma_{22}).
\]

The basic reproduction number of system~\eqref{sys:ode-2p2v} following the formula in equation~\eqref{eq:R0-multi-scpeies} is given by:
\[\mathcal{R}_0^{\eqref{sys:ode-2p2v}}=\max(\mathcal{R}_0^1, \mathcal{R}_0^2)\]
where \(\mathcal{R}_0^1=\rho(\mathcal{B}_1)\) and \(\mathcal{R}_0^2=\rho(\mathcal{B}_2)\), with the matrices \(\mathcal{B}_1\) and \(\mathcal{B}_2\) 
obtained by using a form of a matrix in equation~\eqref{eq:matrix_Bv}, and then
\begin{equation}\label{eq:matrix_Bv_2p2v}
\mathcal{B}_1=\displaystyle \begin{pmatrix}\frac{\beta_{111} \varepsilon_{11}S_1^0}{(\varepsilon_{11} + d_1)(\gamma_{11} + d_1)} & \frac{\beta_{121} \varepsilon_{21} S_1^0}{(\varepsilon_{21} + d_2)(\gamma_{21} + d_2)} \\
\frac{ \beta_{211} \varepsilon_{11} S_2^0}{(\varepsilon_{11} +d_1) (\gamma_{11} + d_1)} & \frac{\beta_{221} \varepsilon_{21}S_{2}^0}{( \varepsilon_{21} + d_2)(\gamma_{21} + d_2)} \end{pmatrix},
\mathcal{B}_2=\displaystyle \begin{pmatrix}
\frac{ \beta_{112} \varepsilon_{12} S_1^0}{( \varepsilon_{12} + d_1)(\gamma_{12} + d_1)} & \frac{\beta_{122} \varepsilon_{22} S_1^0}{(\varepsilon_{22} + d_2)(\gamma_{22} + d_2)} \\
\frac{ \beta_{212} \varepsilon_{12}S_2^0}{(\varepsilon_{12} + d_1)(\gamma_{12} + d_1) } &  \frac{ \beta_{222} \varepsilon_{22}S_2^0}{( \varepsilon_{22} + d_2)(\gamma_{22} + d_2)}
\end{pmatrix} 
\end{equation}

Note that the result provided by using $\R_0^{\eqref{sys:ode-2p2v}}$ does not show the whole picture.
Indeed, one interesting characteristic of \eqref{sys:ODE-multi-species} is that the viruses function in a disconnected way.
This can be inferred from the reducibility of the system discussed in Appendix~\ref{app:normal-form-matrices}.
\begin{theorem}\label{th:exist-mixed-EP}
    Consider \eqref{sys:ode-2p2v} with $\R_{01}>1$ and $\R_{02}\leq 1$.
    Then the DFE of \eqref{sys:ode-2p2v} is unstable and consists of a mixed equilibrium wherein pathogen 1 is present at an endemic level and pathogen 2 is absent. Stability of the pathogen-2-free equilibrium is global and asymptotic with respect to the pathogen-2 subsystem.
\end{theorem}
The existence part of the proof of this result is shown in Appendix~\ref{app:existence-mixed-EP}. 
Global asymptotic stability of the pathogen-2-free equilibrium follows directly from Theorem~\ref{th:GAS_R0} and reducibility of the system.

\subsection{Branching process approximation}

Let $Z=(L_{11},  L_{12},L_{21},  L_{22},  I_{11},  I_{12},I_{21},  I_{22})$ be the multitype branching process approximation of CTMC $X^{2,2}(t)$ with infected types $\ell_{11}$,  $\ell_{12}$, $\ell_{21}$, $\ell_{22}$, $i_{11}$, $i_{12}$, $i_{21}$ and $i_{22}$. 
The p.g.f. for $\bu=(u_{11}^\ell,u_{12}^\ell,u_{21}^\ell,u_{22}^\ell,u_{11}^i,u_{21}^i,u_{21}^i,u_{22}^i)$ is defined as
\begin{equation}\label{eq:pdf-2p2v}
 \bF(\bu)=(f_{11}^\ell(\bu),f_{12}^\ell(\bu),f_{21}^\ell(\bu),f_{22}^\ell(\bu),f_{11}^i(\bu),f_{21}^i(\bu),f_{21}^i(\bu),f_{22}^i(\bu)),
\end{equation}
where, for $p,v=1,2$,
\begin{subequations}\label{sys:pgf-2p2v}
    \begin{align}
        f_{pv}^\ell(\bu) &= \frac{\varepsilon_{pv}u_{pv}^i+d_p}{\varepsilon_{pv}+d_p}, \label{sys:pgf-2p2v_flpq} \\
        f_{pv}^i(\bu) &= \frac{\left(\sum_{q=1}^{2} \beta_{qpv}S_{q}^0 u_{pv}^\ell\right)u_{pv}^i+\gamma_{pv}+d_p}{\Lambda_{pv}}.
        \label{sys:pgf-2p2v_fipq}
    \end{align}
\end{subequations}
The Jacobian matrix then takes the form
\[
DF(u)=\begin{bmatrix}
\begin{array}{c:c}
0 & \mathbb{M}_{12} \\
\hdashline
\mathbb{M}_{21} & \mathbb{M}_{22} \\
\end{array}
\end{bmatrix},
\]
where
\[
\mathbb{M}_{12}= 
\diag\left(
\frac{\varepsilon_{11}}{  \varepsilon_{11}+d_1},
\frac{\varepsilon_{12}}{  \varepsilon_{12}+d_1},
\frac{\varepsilon_{21}}{  \varepsilon_{21}+d_2},
\frac{\varepsilon_{22}}{  \varepsilon_{22}+d_2}
\right),
\]
\[
\mathbb{M}_{21}= \begin{pmatrix}
    \frac{\beta_{111} S_1^0 u_{11}^i}{\Lambda_{11}} & 0 & \frac{ \beta_{121} S_2^0u_{11}^i}{\Lambda_{11}} & 0 \\
0 & \frac{ \beta_{112}S_1^0 u_{12}^i}{\Lambda_{12}} & 0 & \frac{\beta_{122}S_2^0  u_{12}^i}{\Lambda_{12}} \\
\frac{ \beta_{211}S_1^0 u_{21}^i}{\Lambda_{21}} & 0 & \frac{S_2^0 \beta_{221} u_{21}^i}{\Lambda_{21}} & 0 \\
0 & \frac{ \beta_{212}S_1^0 u_{22}^i}{\Lambda_{22}} & 0 & \frac{ \beta_{222}S_2^0 u_{22}^i}{\Lambda_{22}}
\end{pmatrix}
\]
and
\begin{multline*}
\mathbb{M}_{22} = 
\diag
\Biggl(
\frac{ \beta_{111}S_1^0 u_{11}^\ell + \beta_{121} S_2^0 u_{21}^\ell}{\Lambda_{11}},
\frac{ \beta_{112} S_1^0u_{12}^\ell +  \beta_{122}S_2^0 u_{22}^\ell}{\Lambda_{12}},\\
\frac{ \beta_{211}S_1^0 u_{11}^\ell +  \beta_{221}S_2^0 u_{21}^\ell}{\Lambda_{21}},
\frac{ \beta_{212}S_1^0 u_{12}^\ell +\beta_{222}  S_2^0 u_{22}^\ell}{\Lambda_{22}}
\Biggr).  
\end{multline*}

Theorem~\ref{th:exists-q-multi-pv} applies here.
Given $L_{11}(0)=\ell_{110}$, $L_{12}(0)=\ell_{120}$, $L_{21}(0)=\ell_{210}$, $L_{22}(0)=\ell_{220}$, $I_{11}(0)=i_{110}$, $I_{12}(0)=i_{120}$, $I_{21}(0)=i_{210}$ and $I_{22}(0)=i_{220}$, it follows from the independent branching process approximation that the probabilities of extinction and disease outbreak are:
\begin{equation*}
  \mathbb{P}_{\text{ext}}^{\eqref{eq:pdf-2p2v}}=  \left\{
    \begin{aligned}
   & (z_{11}^\ell)^{\ell_{110}}(z_{12}^\ell)^{\ell_{120}}(z_{21}^\ell)^{\ell_{210}}(z_{22}^\ell)^{\ell_{220}} (z_{11}^i)^{i_{110}}(z_{12}^i)^{i_{120}}(z_{21}^i)^{i_{210}}(z_{22}^i)^{i_{220}}  ,  &\mathcal{R}_0^{\eqref{sys:ode-2p2v}}>1\\
    & 1, &\mathcal{R}_0^{\eqref{sys:ode-2p2v}}<1,  
    \end{aligned}
    \right. 
\end{equation*}
    \begin{equation}\label{eq:proba-outbreak-2p2v}
    \mathbb{P}_{\text{outbreak}}^{\eqref{eq:pdf-2p2v}}=  1-\mathbb{P}_{\text{ext}}^{\eqref{eq:pdf-2p2v}}
\end{equation}

\section{Discussion}
\label{sec:discussion}
We formulated an SLIR model for the spread of $V$ pathogens between and within $P$ species.
The model was first formulated as a set of $2P(1+V)$ ordinary differential equations, of which we considered elementary properties: basic reproduction number $\R_0$ and local asymptotic stability as a function of the value of the reproduction number.
Global asymptotic stability when $\R_0<1$ was also established.
We then considered the corresponding continuous-time Markov chain (CTMC) model, which provides a better tool to study the behaviour of the system close to the disease-free equilibrium (DFE), which is our main interest here.
To do so, we employed a multitype branching process approximation of the CTMC near the DFE, obtaining an expression for the probability of an outbreak when $\R_0>1$.
This probability was interpreted, in the context of our model, as the probability that the pathogen becomes established in the population, at least temporarily. 
(The result is local and does not address the proper establishment at an endemic level.)

The case of a single pathogen spreading between two species was then investigated computationally. 
A metapopulation (spatial) version of this situation with a slightly different underlying SLIR model was investigated both mathematically and computationally in \cite{ArinoDavisHartleyJordanEtAl2005}, with, using the notation here, $P$ species present. 
However, focusing on just two species as we did here allows to get a better understanding of the processes.
To this effect, in particular, we investigated the sensitivity of the system to its parameters in the case of three viruses affecting fish leading to very different transmission scenarios.
These highlighted in particular the important role played by demographic parameters.
A fourth, more abstract case concerned introduction of a pathogen in a population by another in which the pathogen is present at an endemic level.

One interesting feature of the model is that despite its complication, pathogens function more or less independently from one another.
This was shown in the case $P=V=2$, which we considered next.
We conjecture that the following natural extension of Theorem~\ref{th:exist-mixed-EP} holds.
\begin{conjecture}\label{conj:exist-mixed-EP}
    Consider \eqref{sys:ODE-multi-species} in which pathogens $a=1,\ldots,k$ have $\R_{0a}<1$ and pathogens $e=k+1,\ldots,V$ have $\R_{0e}>1$, for some $k\in\{1,\ldots,n\}$.
    Then \eqref{sys:ODE-multi-species} has an unstable mixed equilibrium in which pathogens $1,\ldots,k$ are absent and pathogens $k+1,\ldots,V$ are present at an endemic level.
\end{conjecture}

By Theorem~\ref{th:GAS_R0}, those pathogens that are absent would naturally be globally asymptotically stably so.
This highlights a limitation of the model: because we assume that species dynamics is independent of the pathogens and that coinfections cannot occur, a situation as described by Theorem~\ref{th:exist-mixed-EP} or Conjecture~\ref{conj:exist-mixed-EP} is possible.
While not necessarily unrealistic, taking into account more advanced interactions between pathogens, or effects of pathogens onto their host species, could be interesting and lead to wholly different results.
Competition effects between species could also be incorporated and would also likely lead to different results.

Another interesting variation could involve considering an \emph{epidemic} model. 
In its current form, the model is an \emph{endemic} model, with the basic reproduction numbers distinguishing between a situation where the disease is absent and one where the disease is present at an endemic level.
In this case, of course, the pathogen always becomes extinct; it is not clear, though, if a situation similar to that in Theorem~\ref{th:exist-mixed-EP} and Conjecture~\ref{conj:exist-mixed-EP}, in this case with some pathogens undergoing an epidemic and others not, would hold.

Finally, note that we considered an SLIR model, not an SLIRS model, i.e., we assumed that acquired immunity is permanent.
With respect to the work carried out here, this makes very little difference.
Indeed, adding a flow from $R$ to $S$ does not modify the expression of the disease-free equilibrium nor of the basic reproduction number, which is the focus of most of our work.
The only difference between the two model formulations would appear, in the present work in the particular case where the pathogen is endemic in one species (Section~\ref{subsec:intro-from-endemic}).
There, \eqref{eq:EEP_E1} would be slightly modified, but this has no consequence since $E_1^\star$ is used as a parameter.
The expressions in Section~\ref{subsec:ODE-introduction-model} would also change, as some would incorporate the rate of movement from $R$ to $S$, but the overall conclusions remain.



\subsection*{Acknowledgements}
The authors thank Sharon Clouthier, Karen Dunmall and Collin Ghallagher, from the Freshwater Institute (Department of Fisheries and Oceans Canada, Winnipeg) for initiating the project and helpful discussions. 
They also thank anonymous reviewers whose comments greatly improved the manuscript.
JA acknowledges support from NSERC through the Discovery Grants program. This work was partially supported by a 2024/2025 MMI Research Accelerator Award and we gratefully acknowledge the Maud Menten Institute (PRN2).

\appendix 

\section{Normal forms of matrices}
\label{app:normal-form-matrices}
When considering the global asymptotic stability of the DFE in Appendix~\ref{app:proof_GAS_R0} or existence of threshold behaviour in Appendix~\ref{app:proof-exists-q-multi-pv}, we observe that the matrices involved are reducible.
While the presentation used in the body of the paper is the most natural, the proofs in these appendices require to use the normal form of the matrices involved.
Further understanding of the reproduction number is also gained by using the normal form. 
Let us illustrate this using matrix $A_{22}$ in Appendix~\ref{app:proof_GAS_R0}.
The permutation matrix obtained in the process is the same for all four matrices we consider here.

It is clear that matrix $A_{22}$ as given by \eqref{eq:matrix-A22-KSmethod} is reducible.
Indeed, consider the weighted loop-directed graph having $A_{22}$ as its adjacency matrix. 
Diagonal entries of the two main diagonal blocks correspond to loops. 
The blocks $\bm{\Psi}(\bx_S,\bx_I)$ and $\diag(\llbracket\varepsilon_{pv}\rrbracket)$, on the other hand, show that vertices form $PV$ strongly connected components, with each component comprising $2$ vertices.
Consider for instance the $(1,1)$ entry in $\bm{\Psi}$ and the $(1,1)$ entry in $\diag(\llbracket\varepsilon_{pv}\rrbracket)$. They establish that vertex 1 is connected to vertex $PV+1$ and that vertex $PV+1$ is connected to vertex 1. All entries in these matrices define similar pairs of vertices, giving the $PV$ strong components.
Ordering vertices (and corresponding matrix entries) so that vertices in a strong connected component are listed consecutively, it is then easy to find the permutation matrix $\bm{\Pi}$ such that
\begin{equation}\label{eq:A22-block-diag-form}
\bm{\Pi}\ A_{22}(\bx_S,\bx_I) \bm{\Pi}^T=
\bigoplus_{\llbracket pv\rrbracket}
\begin{pmatrix}
-(\varepsilon_{pv}+d_p) & \sum_{q=1}^P\beta_{qpv}S_q \\
\varepsilon_{pv} & -(\gamma_{pv}+d_p)
\end{pmatrix}.
\end{equation}
Applying $\bm{\Pi}$ to $D\bF$ given by \eqref{eq:DF} in Appendix~\ref{app:proof-exists-q-multi-pv}, we find
\begin{equation}\label{eq:DF_block}
\bm{\Pi}\ D\bF(\bu^\ell,\bu^i)\ \bm{\Pi}^T=
\bigoplus_{\llbracket pv\rrbracket}
\begin{pmatrix}
0 & \dfrac{\varepsilon_{pv}}{\varepsilon_{pv}+d_p} \\
S_p^0\dfrac{\beta_{pqv}u_{pv}^i}{\Lambda_{pv}} &
\dfrac{\sum_{q=1}^PS_q^0u_{qv}^\ell}{\Lambda_{pv}}
\end{pmatrix}.
\end{equation}

\section{Proof of Theorem~\ref{th:GAS_R0}}
\label{app:proof_GAS_R0}
To prove global asymptotic stability of the disease-free equilibrium of \eqref{sys:ODE-multi-species} when $\R_0<1$, we use a result of Kamgang and Sallet \cite[Theorems 4.3 and 4.5]{KAMGANG20081}.
For the convenience of the reader, we recall this result but with notation adapted to the problem under consideration here.

Let $\bx_S=(\llbracket S_p\rrbracket,\llbracket R_p\rrbracket)^T\in\IR^{2P}$ and $\bx_I=(\llbracket L_{pv}\rrbracket,\llbracket I_{pv}\rrbracket)^T\in\IR^{2PV}$ be, respectively, the vectors of non-infected compartments and infected compartments. 
Denote $\bx_S^0 =(\llbracket S_p^0\rrbracket,\llbracket 0\rrbracket)^T\in\IR^{2P}$ the part of the disease-free equilibrium of \eqref{sys:ODE-multi-species} corresponding to $\bx_S$, i.e., the DFE is $(\bx_S^0,\b0_{2PV})$.
Rewrite \eqref{sys:ODE-multi-species} in the following compact form,
\begin{equation}\label{sys:com-multi-species}
    \begin{aligned}
    \dot{\bx}_S &= A_{11}(\bx_S,\bx_I)(\bx_S-\bx_S^0) + A_{12}(\bx_S,\bx_I)\bx_I, \\
    \dot{\bx}_I &= A_{22}(\bx_S,\bx_I)\bx_I,
\end{aligned}
\end{equation}
with
\begin{align}
A_{11}(\bx_S, \bx_I)
&=
\begin{pmatrix}
-\diag\left(\llbracket d_p\rrbracket\right) & \b0 \\
\b0 & -\diag\left(\llbracket d_p\rrbracket\right)\\
\end{pmatrix}\in\IR^{2P\times 2P},  \nonumber
\\  
A_{12}(\bx_S, \bx_I)
&=
\begin{pmatrix}
\b0 & -\bm{\Phi}(\bx_S,\bx_I) \\
\b0 & \bm{\Gamma}(\bx_S,\bx_I) \\
\end{pmatrix}\in\IR^{2P\times 2PV}, \nonumber \\
A_{22}(\bx_S,\bx_I) &=
\begin{pmatrix}
    -\diag\left(\llbracket\varepsilon_{pv}+d_p\rrbracket\right) & 
    \bm{\Psi}(\bx_S,\bx_I) \\
    \diag\left(\llbracket\varepsilon_{pv}\rrbracket\right) &
    -\diag\left(\llbracket\gamma_{pv}+d_p\rrbracket\right)
\end{pmatrix}\in\IR^{2PV\times 2PV},
\label{eq:matrix-A22-KSmethod}
\end{align}
where $\Phi(\bx_S,\bx_I)$ is a $P\times PV$-matrix whose $j$th row is a $PV$-vector with entries $\llbracket S_j\beta_{jpv}\rrbracket$, $\bm{\Gamma}(\bx_S,\bx_I)$ is a $P\times PV$-matrix with $j$th row having $V$ nonzero entries $\gamma_{j1},\ldots,\gamma_{jV}$ in columns $(j-1)V+1$ to $jV$ and $\bm{\Psi}(\bx_S,\bx_I)$ is a diagonal $PV\times PV$-matrix,
\[
\bm{\Psi}(\bx_S,\bx_I) = \diag\left(\left\llbracket
\sum_{q=1}^P\beta_{qpv}S_q
\right\rrbracket\right),
\]
with the enumerator running over indices $p=1,\ldots,P$ and $v=1,\ldots,V$.

As discussed in Appendix~\ref{app:normal-form-matrices}, it is clear that matrix $A_{22}$ as given by \eqref{eq:matrix-A22-KSmethod} is reducible.
So we apply the method described in that Appendix and instead work with the similar normal form matrix \eqref{eq:A22-block-diag-form}.
We can apply \cite[Theorem 4.3]{KAMGANG20081} to each of the $PV$ blocks in \eqref{eq:A22-block-diag-form} or apply \cite[Theorem 4.5]{KAMGANG20081}, which considers the reducible case. 
We use a combination of the two results.

Let $\Omega \subset \mathbb{R}^{2P}_+ \times \mathbb{R}^{2PV}_+$ be the set defined in the statement of Theorem~\ref{th:GAS_R0}. 
If the following five conditions hold true, then \cite[Theorem 4.5]{KAMGANG20081} establishes that the disease-free equilibrium $(\bx_S^0,\b0_{2PV})$ is globally asymptotically stable when $\R_0<1$.
\begin{description}
    \item[C1] System~\eqref{sys:com-multi-species} is defined on a positively invariant set $\Omega$ of the nonnegative orthant and dissipative on $\Omega$. 
    \item[C2] Subsystem $\dot{\bx}_S = A_{11}(\bx_S,\b0)(\bx_S - \bx_S^0)$ is globally asymptotically stable at the disease-free equilibrium $\bx_S^0$ on the canonical projection of $\Omega$ on $\mathbb{R}^{2P}_+$. 
    \item[C3] Matrix $\tilde{A}_{22}(\bx_S,\bx_I)$ given by \eqref{eq:A22-block-diag-form} is block upper triangular, with each diagonal block $\hat{A}_{22}^{pv}$ being Metzler and irreducible for any $\bx=(\bx_S,\bx_I)\in\Omega$.
    \item[C4] For each $p,v$, there exists an upper-bound matrix $\hat{A}_{22}^{pv}$ for $M = \{A_{22}^{pv}(\bx) \in \IR^{2\times 2};\bx \in \Omega\}$ with the property that either $\hat{A}_{22}^{pv} \not\in M$ or if $\hat{A}_{22}^{pv} \in M$, (i.e., $\hat{A}_{22} = \max_{\Omega} M$), then for any $\hat{\bx} \in \Omega$ such that $\hat{A}_{22}= A_{22}(\hat{\bx})$, $\hat{\bx} \in \mathbb{R}^{2P}_+ \times \{0\}^{2PV}$.
    \item[C5] The spectral abscissa of the matrix $\hat{A}_{22}^{pv}$ verifies $\sigma(\hat{A}_{22}^{pv}) \leq 0$ when $\mathcal{R}_0^{\eqref{sys:ODE-multi-species}}\leq 1$.
\end{description}
The right-hand side of \eqref{sys:com-multi-species} is of class $C^1$ on the open set $\mathbb{R}^{2P}_+ \times \mathbb{R}^{2PV}_+$, so solutions are defined.
Solutions in $\Omega$ remain in $\Omega$.
Furthermore, extending $\Omega$ as
\begin{multline*}
    \Omega_\delta = \Biggl\{
    (\llbracket S_p\rrbracket,\llbracket L_{pv}\rrbracket,\llbracket I_{pv}\rrbracket,\llbracket R_p\rrbracket) \in \mathbb{R}^{2P(V+1)} : \\
    N_p=S_P+\sum_{v=1}^{V}(L_{pv}+I_{pv})+R_p \leq \frac{b_p}{d_p}+\delta; \quad p=1,\ldots, P
    \Biggr\},
\end{multline*}
we have that solutions with initial conditions in $\IR^{2P(V+1)}$ eventually enter $\Omega_\delta$, for any $\delta>0$.
As a consequence, \eqref{sys:com-multi-species} is dissipative and \textbf{C1} is satisfied.
Then note that in the original form, the model without disease is, for $p=1,\ldots,P$,
\begin{subequations}\label{sys:EDO-multi-species-no-disease}
    \begin{align}
        \dot S_p &= b_p-d_pS_p, \label{sys:multi-species-no-disease_dS} \\
        \dot R_p &= -d_pR_{pv}. \label{sys:multi-species-no-disease_dR} 
    \end{align}
\end{subequations}
Thus it is clear that for each $p=1,\ldots,P$, $S_p(t)\to b_d/d_p$ and $R_p(t)\to 0$ regardless of initial conditions, meaning that condition \textbf{C2} holds.
Matrix \eqref{eq:A22-block-diag-form} is in block-diagonal form and as a consequence, is block upper triangular. Recall that a matrix is Metzler if its offdiagonal entries are nonnegative, so \textbf{C3} holds.

To verify that \textbf{C4} holds, remark that for solutions in $\Omega$, the maximal value for a given $S_p$ is attained when $S_p=N_p$, i.e., at the disease-free equilibrium $\bx_S^0$. 
In other words, for $\tilde{\bx} \in \mathbb{R}^{2P}_+ \times \{0\}^{2PV}$, $\tilde{\bx}_S \leq\bx_S^0$. 
Thus, the upper-bound  matrix is $\hat{A}_{22} = \hat{A}_{22}(\bx_S^0,\b0)=\mathbf{G}-\mathbf{W}$.
Note that the result can also be formulated using blocks $\hat{A}_{22}^{pv}$ as in \eqref{eq:A22-block-diag-form} and selecting the relevant components in $\bm{\Pi}(\mathbf{G}-\bW)\bm{\Pi}^T$, but we do need the form using the unreduced matrix $\hat{A}_{22}$ to show that condition $\textbf{C5}$ holds.

Indeed, to show \textbf{C5}, return to the unreduced form \eqref{eq:matrix-A22-KSmethod} and note that we have $\sigma(\hat{A}_{22}) = \sigma(\mathbf{G} - \mathbf{W})$, with the matrices as defined in Section~\ref{sec:basic_analysis_ODE} and satisfying the conditions of \cite[Theorem 2]{VdDWatmough2002}.
The proof of that theorem establishes that $\R_0<1$, i.e., the spectral radius $\rho(\mathbf{G}\bW^{-1})<1$, is equivalent to the spectral abscissa $\sigma(\mathbf{G}-\bW)<0$.
As a consequence, when $\R_0<1$, then $\sigma(\hat{A}_{22})<0$ and the same is true for each diagonal block $\hat{A}_{22}^{pv}$ in the matrix in normal form \eqref{eq:A22-block-diag-form}, since it is similar to \eqref{eq:matrix-A22-KSmethod}.

Since conditions \textbf{C1}--\textbf{C5} are satisfied, the proof is done.
\section{Proof of threshold behaviour in Theorem~\ref{th:exists-q-multi-pv}}
\label{app:proof-exists-q-multi-pv}
To continue the proof of Theorem~\ref{th:exists-q-multi-pv}, we now need to consider the existence of fixed points of the p.g.f. $\bF$.
The result will follow from application of the Threshold Theorem \cite{ALLEN201399} together with \cite[Theorem 7.1]{Harris1963}, provided we can show that the branching process is positive and regular.

The Jacobian of~\eqref{sys:pgf-multi-species}  is  the $2PV\times 2PV$ block matrix
\begin{equation}\label{eq:DF}
D\bF\left(\bu^\ell,\bu^i\right)
=
\begin{pmatrix}
\begin{array}{cc}
\bm{0} & \mathbb{M}_{12} \\
\mathbb{M}_{21} & \mathbb{M}_{22} \\
\end{array}
\end{pmatrix},    
\end{equation}
where each block has size $PV\times PV$. First,
\[
\mathbb{M}_{12}=\diag\left(\left\llbracket\frac{\varepsilon_{pv}}{\varepsilon_{pv}+d_p}\right\rrbracket\right),\; 
\mathbb{M}_{22}=\diag\left(\left\llbracket\frac{\sum_{q=1}^{P}S_{q}^0 u_{qv}^\ell}{\Lambda_{pv}}\right\rrbracket\right).
\]
Then, the $PV\times PV$-matrix $\mathbb{M}_{21}$ is itself a block matrix, with each $V\times V$ sized block taking the form, for $p,q\in\{1,\ldots,P\}$,
\[
\mathcal{K}_{pq}=\displaystyle S_{q}^0\diag\left(\left\llbracket \frac{\beta_{pqv}u_{pv}^i}{\Lambda_{pv}} \right\rrbracket\right).
\]

As established in Appendix~\ref{app:normal-form-matrices}, \eqref{eq:DF} is reducible.
Using the similarity transformation in Appendix~\ref{app:normal-form-matrices}, we derive the normal form of $D\bF$, \eqref{eq:DF_block}.
As in the proof in Appendix~\ref{app:proof_GAS_R0}, it is useful to use both the original matrix \eqref{eq:DF} and its normal form \eqref{eq:DF_block}.

\begin{enumerate}
\item [(i)]
It is clear that \eqref{eq:DF} is such that if $\bx,\by\in[0,1)^{2PV}$ are such that $\bx\leq\by$, one has $D\bF(\bx)\leq D\bF(\by)$.
Indeed, terms $\bu^\ell$ and $\bu^i$ appear as sums in the numerators of the expressions involving them.
Furthermore, $\bF(\b0)>\b0$, i.e., it is a nonnegative matrix with some positive entries.
This implies that the multitype branching processes are not singular \cite[Theorem 2.3]{berman1979nonnegative}.

\item[(ii)] The matrix of first moments is $\mathbb{M}=D\bF(\mathds{1}_{2PV})$, where
\(\mathds{1}_{2PV}\) is the unit column vector of size $2PV$.
In the transformed matrix \eqref{eq:DF_block}, diagonal blocks of $\mathbb{M}$ take the form
\[
\begin{pmatrix}
0 & \dfrac{\varepsilon_{pv}}{\varepsilon_{pv}+d_p} \\
S_p^0\dfrac{\beta_{pqv}}{\Lambda_{pv}} &
\dfrac{\sum_{q=1}^PS_q^0}{\Lambda_{pv}}
\end{pmatrix}
\]
and are therefore irreducible (and even primitive).
Consequently, the matrix of first moments $\mathbb{M}$ is block-primitive.
\end{enumerate}
From $(i)$ and $(ii)$, we conclude that the branching process is positive and regular.
As a consequence, applying the Threshold Theorem \cite{ALLEN201399} together with \cite[Theorem 7.1]{Harris1963} to each of the diagonal blocks in the matrix in normal form \eqref{eq:DF_block}, gives the threshold behaviour, with existence of a fixed point $(0,0)<(z_{pv}^\ell,z_{pv}^i)<(1,1)$ additionally to $(z_{pv}^\ell,z_{pv}^i)=(1,1)$ when the process is supercritical. 
Putting things together, under the conditions of Theorem~\ref{th:exists-q-multi-pv}, there exists an additional fixed point $\b0<\bz<\bm{1}$ when $\R_0>1$.

\section{Existence of a mixed equilibrium for \eqref{sys:ode-2p2v}}
\label{app:existence-mixed-EP}
Suppose that the conditions of Theorem~\ref{th:exist-mixed-EP} are satisfied: $P=V=2$, virus 2 is at the disease-free equilibrium (DFE), i.e., $L_{12}=L_{22}=I_{12}=I_{22}=0$ and $\R_{02}<1$.
We seek equilibria of \eqref{sys:ode-2p2v} with positive values for $L_{11}^\star$, $I_{21}^\star$, $I_{11}^\star$ and $I_{21}^\star$, under the assumption that $\R_{01}>1$.

Substituting the DFE of species 2 into \eqref{sys:ode-2p2v}, we obtain from \eqref{sys:ode-2p2v_Lq1} that
\[
S_1^\star=\frac{\varepsilon_{11}+d_1}{\beta_{111} I_{11}^\star+\beta_{121} I_{21}^\star} L_{11}^\star, \quad S_2^\star=\frac{\varepsilon_{21}+d_2}{\beta_{211} I_{11}^\star+\beta_{221} I_{21}^\star}L_{21}^\star,
\]
while \eqref{sys:ode-2p2v_Iq1} gives
\[
L_{11}^\star=\frac{\gamma_{11}+d_1}{\varepsilon_{11}} I_{11}^\star, \quad L_{21}^\star=\frac{\gamma_{21}+d_2}{\varepsilon_{21}} I_{21}^\star
\]
and, finally, from \eqref{sys:ode-2p2v_R},
\[
R_1^\star=\frac{\gamma_{11}}{d_1} I_{11}^\star, \quad R_2^\star=\frac{\gamma_{21}}{d_2} I_{21}^\star.
\]
Since the total population of each species is governed, for $i=1,2$, by $\dot N_i=b_i-d_iN_i$, at an equilibrium, $b_i-d_iN_1^\star=0$ and thus, when \eqref{sys:ode-2p2v} is at an equilibrium with virus 2 at the DFE, one has
\[  
b_1-d_1\left(S_1^\star+L_{11}^\star+I_{11}^\star+R_1^\star\right)=0 
\quad \text{and} \quad 
b_2-d_2\left(S_2^\star+L_{21}^\star+I_{21}^\star+R_2^\star\right)=0.
\]

Expressing all terms as functions of $I_{i1}$, $i=1,2$, and using the expressions of $S_1$, $L_{11}$, $R_1$, $S_2$, $L_{21}$ and $R_2$ gives
\begin{subequations}\label{sys:sol_ode_I}
    \begin{align}
        \beta_{111} (I_{11}^\star)^2-\left[d_1\left(\mathcal{R}_{01}-1\right)-\beta_{121}I_{21}^\star\right] I_{11}^\star-\frac{\beta_{121}}{\beta_{111}} d_1 \mathcal{R}_{01} I_{21}^\star=0,\label{sys:sol_ode_I11}\\
        \beta_{221} (I_{21}^\star)^2-\left[d_2\left(\mathcal{R}_{02}-1\right)-\beta_{211} I_{11}^\star\right] I_{21}^\star-\frac{\beta_{211}}{\beta_{221}} d_2 \mathcal{R}_{02} I_{11}^\star=0. \label{sys:sol_ode_I21}
    \end{align}
\end{subequations}

To simplify computations, let us denote $x=I_{11}^\star$ and $y=I_{21}^\star$. Then \eqref{sys:sol_ode_I} can be written as
\begin{subequations}\label{sys:sol_ode_I_xy}
    \begin{align}
        \Gamma_1(x,y) &:= 
        \beta_{111}x^2+\beta_{121}xy
        -d_1\left(\mathcal{R}_{01}-1\right)x
        -\frac{\beta_{121}}{\beta_{111}} d_1 \mathcal{R}_{01}y=0,
        \label{sys:sol_ode_I_xy_x}\\
        \Gamma_2(x,y) & :=
        \beta_{211} xy
        +\beta_{221} y^2
        -\frac{\beta_{211}}{\beta_{221}} d_2 \mathcal{R}_{02} x
        -d_2\left(\mathcal{R}_{02}-1\right)y=0. 
        \label{sys:sol_ode_I_xy_y}
    \end{align}
\end{subequations}
Both $\Gamma_1$ and $\Gamma_2$ are conic sections.
Their discriminants $-\beta_{121}^2$ and $-\beta_{211}^2$ are both negative, hence they are both hyperbolas.
As \eqref{sys:sol_ode_I_xy_x} has no second degree $y$ monomial, one of its asymptotes is vertical. 
Likewise, since \eqref{sys:sol_ode_I_xy_y} has no second degree $x$ monomial, one of its asymptotes is horizontal.
If both of these asymptotes intersect the positive (or first) quadrant $Q_1=\IR_+\times\IR_+$, then $\Gamma_1$ and $\Gamma_2$ intersect a single time there.
The situation is shown in Figures~\ref{fig:Gamma1-Gamma2} and \ref{fig:Gamma1-Gamma2-zoom}.

\begin{figure}[htbp]
    \centering
    \includegraphics[width=0.65\textwidth]{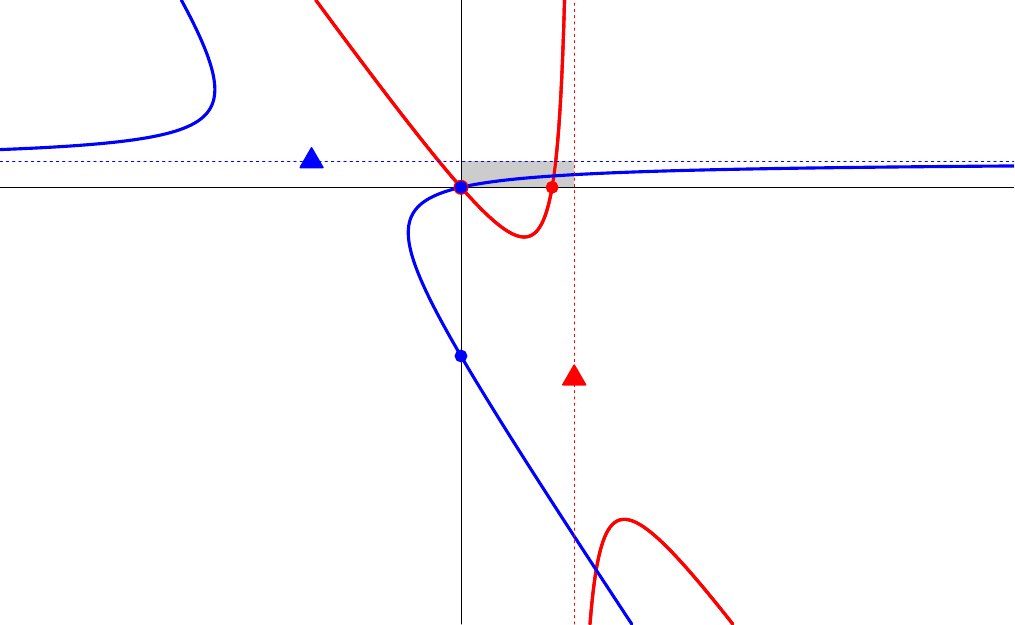}
    \caption{Situation leading to the existence of a mixed equilibrium. Red curve: $\Gamma_1$. Blue curve: $\Gamma_2$. The centres of the hyperbolas are shown as triangles of corresponding colours, as are the vertical and horizontal asymptotes relevant to the problem. The shaded box shows the possible range of values of the endemic component of the mixed equilibrium.}
    \label{fig:Gamma1-Gamma2}
\end{figure}

Let us show that this is indeed the case.
There are two ingredients:
\begin{enumerate}
    \item The curves $\Gamma_1$ and $\Gamma_2$ intersect $Q_1$.
    \item The vertical asymptote of $\Gamma_1$ and the horizontal asymptote of $\Gamma_2$ intersect $Q_1$.
\end{enumerate}

First, note that the point of intersection $(x,y)=(0,0)$ is obvious since neither \eqref{sys:sol_ode_I_xy_x} nor \eqref{sys:sol_ode_I_xy_y} have terms of degree 0. 
Now consider the $x$- and $y$-intercepts of $\Gamma_1$ and $\Gamma_2$. 
For $\Gamma_1$, $x$-intercepts satisfy
\[
\Gamma_1(x,0) = \left(\beta_{111}x-d_1(\R_{01}-1)\right)x,
\]
i.e., $x=0$ and $x=d_1(\R_{01}-1)/\beta_{111}>0$ by the assumption $\R_{01}>1$. 
For $y$-intercepts,
\[
\Gamma_1(0,y) = -\frac{\beta_{121}}{\beta_{111}}d_1\R_{01}y,
\]
i.e., $\Gamma_1$ intercepts the $y$-axis only at the origin.
For $\Gamma_2$, $x$-intercepts are given by
\[
\Gamma_2(x,0) = -\frac{\beta_{211}}{\beta_{221}}d_2\R_{02}x,
\]
giving only the origin as $x$-intercept, while $y$-intercepts are given by
\[
\Gamma_2(0,y) = \left(\beta_{221}y-d_2(\R_{02}-1)\right)y,
\]
i.e., $y$-intercepts are $y=0$ and $y=d_2(\R_{02}-1)/\beta_{221}<0$ by assumption.

We can therefore establish that $\Gamma_1$ and $\Gamma_2$ intersect $Q_1$. Indeed, from the gradients $\nabla\Gamma_1(x,y)$ and $\nabla\Gamma_2(x,y)$, 
we deduce that vectors tangent to $\Gamma_1$ and $\Gamma_2$ are
\begin{align*}
    T_1(x,y)=\left(
        \beta_{121}x -\frac{\beta_{121}}{\beta_{111}}d_1\R_{01},
        d_1(\R_{01}-1)-2\beta_{111}x-\beta_{121}y
    \right), \\
    T_2(x,y)\left( 
        2\beta_{221}y +\beta_{211}x -d_2(\R_{02}-1),
        \frac{\beta_{211}}{\beta_{221}}d_2\R_{02}-\beta_{211}y
    \right).
\end{align*}
At the origin, $T_1(0,0)=(-\beta_{121}d_1\R_{01}/\beta_{111},d_1(\R_{01}-1))$ has signs $(-,+)$. This means that $\Gamma_1$ ``moves'' through origin from the second quadrant $Q_2=\IR_-\times\IR_+$ when $x<0$ to the fourth quadrant $Q_4=\IR_+\times\IR_-$ when $x>0$.
On the other hand, $T_2(0,0)=(-d_2(\R_{02}-1),\beta_{211}d_2\R_{02}/\beta_{221})$ has signs $(+,+)$, implying that left of the $y$-axis, $\Gamma_2$ is in the third quadrant $Q_3=\IR_-\times\IR_-$, while it is in $Q_1$ when $x>0$.

It remains to show that the vertical and horizontal asymptotes of $\Gamma_1$ and $\Gamma_2$, respectively, intersect the positive quadrant $Q_1$.

The centre of $\Gamma_1$ is $(d_1\R_{01}/\beta_{111},-d_1(\R_{01}+1)/\beta_{121})\in Q_4$ (red triangle in Figure~\ref{fig:Gamma1-Gamma2}).
Since $\Gamma_1$ intersects $Q_2$, the asymptote to $\Gamma_1\cap Q_2$ has negative slope. 
It follows that the vertical asymptote to $\Gamma_1$ is the one to $\Gamma_1\cap Q_1$ and therefore, intersects $Q_1$.
Reasoning similarly, observe that since the centre $(-d_2(\R_{02}+1)/\beta_{211},d_2\R_{02}/\beta_{221})$ of $\Gamma_2$ lies in $Q_2$ (blue triangle in Figure~\ref{fig:Gamma1-Gamma2}) and $\Gamma_2$ intersects $Q_4$, the asymptote to $\Gamma_2\cap Q_4$ has negative slope and that to the part of $\Gamma_2\cap Q_1$ is horizontal.

As a consequence, there is a point in the interior of the positive quadrant $Q_1$ where $\Gamma_1$ intersects $\Gamma_2$.
More precisely, remark that since the centres of the hyperbola lie on the vertical and horizontal asymptotes, respectively, the $x$-coordinate of the point of intersection cannot be larger that the $x$-component of the centre of $\Gamma_1$ and its $y$-coordinate cannot exceed the $y$-component of the centre of $\Gamma_2$.
This means that the endemic equilibrium belongs to the box
\begin{equation}
    \left(
    0,
    \frac{d_1\R_{01}}{\beta_{111}}
    \right]\times
    \left(0,
    \frac{d_2\R_{02}}{\beta_{221}}
    \right]
\end{equation}
shown shaded in Figures~\ref{fig:Gamma1-Gamma2} and~\ref{fig:Gamma1-Gamma2-zoom}.

\begin{figure}[htbp]
    \centering
    \includegraphics[width=0.65\textwidth]{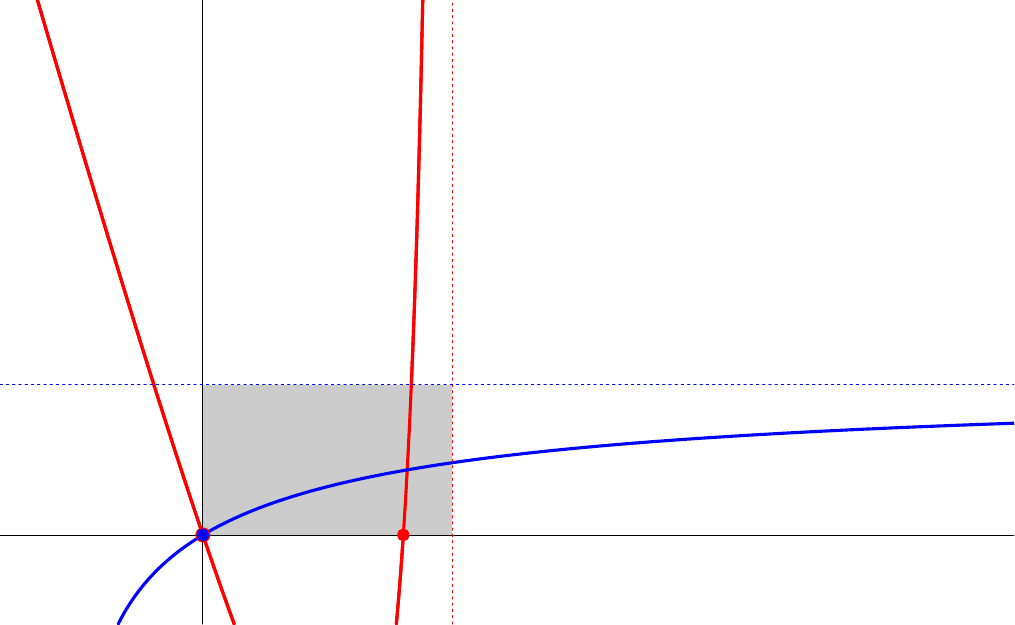}
    \caption{Situation leading to the existence of a mixed equilibrium. Zoom on Figure~\ref{fig:Gamma1-Gamma2} focusing on the positive quadrant. $\Gamma_1$ and its related features is shown in red, while $\Gamma_2$ is in blue.}
    \label{fig:Gamma1-Gamma2-zoom}
\end{figure}

\bibliographystyle{plain}
\bibliography{references}

\end{document}